\documentclass[11pt]{amsart}
\usepackage[utf8]{inputenc}
\usepackage[nospace,noadjust]{cite}
\usepackage{relsize}
\usepackage{cite}
\usepackage{color}
\usepackage[dvipsnames]{xcolor}
\usepackage{tikz}
\usetikzlibrary{arrows.meta, positioning}
\usepackage{tikz-cd}
\usetikzlibrary{arrows.meta, calc, decorations.markings}
\usepackage{colortbl}
\usepackage{hyperref, enumitem}
\hypersetup{
  colorlinks   = true, %Colours links instead of ugly boxes
  urlcolor     = blue, %Colour for external hyperlinks
  linkcolor    = Purple, %Colour of internal links%
  citecolor   = red %Colour of citations
}
\usepackage[margin=2.5cm]{geometry}

\usepackage{comment}
\usepackage{hyperref}
\usepackage{cleveref,amssymb}

\usepackage{amscd}

\DeclareMathOperator{\C}{\mathcal{C}}
\DeclareMathOperator{\V}{\mathbb{V}}

\newtheorem{theorem}{Theorem}[section]

\newtheorem{lemma}[theorem]{Lemma}
\newtheorem{example}[theorem]{Example}
\newtheorem{corollary}[theorem]{Corollary}
\newtheorem{definition}[theorem]{Definition}
\newtheorem{proposition}[theorem]{Proposition}

\newtheorem{remark}[theorem]{Remark}

\newtheorem{construction}[theorem]{Construction}

\usepackage{url}

\newcommand{\fqn}{\mathbb{F}_{q^n}}
\newcommand{\fqk}{\mathbb{F}_{q^k}}

\newcommand{\cA}{{\mathcal A}}

\newcommand{\cC}{{\mathcal C}}

\newcommand{\cG}{{\mathcal G}}

\newcommand{\cP}{{\mathcal P}}

\newcommand{\F}{{\mathbb F}}

\newcommand{\fq}{{\mathbb F}_{q}}

\newcommand{\N}{\mathrm{N}}

\newcommand{\Orb}{\mathrm{Orb}}
\newcommand{\im}{\mathrm{Im}}

\title{Asymptotically optimal cyclic subspace codes}

\author{Chiara Castello and Paolo Santonastaso}

\date{}

\begin{document}

\begin{abstract}
Subspace codes, and in particular cyclic subspace codes, have gained significant attention in recent years due to their applications in error correction for random network coding. In this paper, we introduce a new technique for constructing cyclic subspace codes with large cardinality and prescribed minimum distance. Using this new method, we provide new constructions of cyclic subspace codes in the Grassmannian $\mathcal{G}_q(n,k)$ of all $k$-dimensional $\F_q$-subspaces of an $n$-dimensional vector space over $\fq$, when $k\mid n$ and $n/k$ is a composite number, with minimum distance $2k-2$ and large size. We prove that the resulting codes have sizes larger than those obtained from previously known constructions with the same parameters. Furthermore, we show that our constructions of cyclic subspace codes asymptotically reach the Johnson type bound II for infinite values of $n/k$.
\end{abstract}

\maketitle
{\textbf{Keywords}: Subspace code; cyclic subspace code; Sidon space}

{\textbf{MSC2020}: Primary 11T71. Secondary  11T99; 94B99.}

\section{Introduction}

The work \cite{koetter2008coding} of Koetter and Kschischang sparked significant interest in subspace codes, primarily due to their efficiency in the error-correction in random network coding, a decoding scheme introduced in the seminal paper \cite{ahlswede2001perfect} by Ahlswede, Cai, Li, and Yeung. Since then, extensive research has focused on the theory of subspace codes, aiming to improve bounds, constructions, and decoding algorithms (cf.  \cite{greferath2018network}). %A network channel can be represented as a directed acyclic graph with three types of vertices: sources (no incoming edges), sinks (no outgoing edges), and inner nodes (both incoming and outgoing edges). When inner nodes combine incoming data linearly, it is known as linear network coding, and if the network's topology is unknown, it is referred to as random network coding.

%Koetter and Kschischang \cite{koetter2008coding} established the mathematical foundation for random linear network coding, emphasizing the importance of subspace codes for reliable communication. 
The natural ambient space in which to define a subspace code is an $n$-dimensional vector space $\V$ over $\F_q$, where $\F_q$ denotes the finite field with $q$ elements, where $q$ is a power of a prime. The set $\mathcal{P}_q(\V)$ of all the $\fq$-subspaces of $\V$ forms a metric space when endowed with the \emph{subspace distance}, defined as 
\[
d(U,V)=\dim_{\fq}(U+V)-\dim_{\fq}(U\cap V),
\]
for any $U,V\in \mathcal{P}_q(\V)$. A subspace code $\C$ is a subset of the metric space $(\cP_q(\V),d)$ and as usual the \emph{minimum distance} of $\cC$ is defined as
\[ d(\cC)=\min\{ d(U,V) \colon U,V \in \cC, U\ne V \}. \]  Consider the Grassmannian $\cG_q(n,k)$, i.e. the set of the $k$-dimensional $\fq$-subspaces of $\V$. Subspace codes in $\cG_q(n,k)$ are refereed as \emph{constant-dimension subspace codes}. One of the most studied class is given by \emph{cyclic subspace codes} introduced in \cite{etzion2011error}. Consider $\V=\F_{q^n}$ as $n$-dimensional vector space over its subfield $\F_q$, a code $\cC \subseteq \mathcal{G}_q(n,k)$ is said to be \emph{cyclic} if for every $\alpha \in \F_{q^n}^*$ and every $V \in \cC$ then $\alpha V \in \cC$. This is equivalent to say that 
\[
\cC=\bigcup_{i=1}^s\mathrm{Orb}_{\fqn^*}(V_i)\subseteq \cG_q(n,k)
\]
for some $\fq$-subspaces $V_1,\ldots,V_s\in\cG_q(n,k)$, where $\Orb_{\F_{q^n}}(V_i)=\{\alpha V_i:\alpha \in \F_{q^n}^*\}$ is the orbit of $V_i$ under the action of the Singer cycle defined by the multiplicative group of $\F_{q^n}$. For this reason, $\cC$ is called \emph{cyclic subspace code}. In the case in which $s=1$, and so
\[
\cC=\mathrm{Orb}_{\fqn^*}(V)=\lbrace \alpha V\colon \alpha\in\fqn^*\rbrace,
\]
for some $\fq$-subspace $V$ of $\fqn$, $\cC$ is also called \emph{one-orbit cyclic subspace code} and $V$ is said to be a \emph{representative} of the orbit; whereas if $s>1$, then $\cC$ is also called \emph{multi-orbit cyclic subspace code}. In the next, we write $\mathrm{Orb}(V)$ instead of $\mathrm{Orb}_{\fqn^*}(V)$ when the Singer cycle acting on $V$ is clear from the context.\\
It is well known that if $\mathrm{Orb}(V)\subseteq\cG_q(n,k)$ has maximum possible distance, i.e. $2k$, then $k$ is a divisor of $n$ and $V$ is a shift of the subfield $\fqk$. These codes are known as spread codes (see e.g. \cite{manganiello2008spread}). Their downside is their small size: they contain $(q^n-1)/(q^k-1)$ codewords, which is the smallest size of any cyclic subspace code generated by a $k$-dimensional subspace. On the other hand, the largest size of a one-orbit cyclic subspace code is $(q^n-1)/(q-1)$, and codes attaining this size are said to be \emph{full-length}. 
%Full-length cyclic subspace codes with minimum distance $2k-2$, which is the best possible situation in terms of code size and minimum distance, are called \emph{optimal (full-length) cyclic subspace codes}.\\
In \cite{trautmann2013cyclic} the authors conjectured the existence of full length one-orbit cyclic subspace codes in $\cG_q(n,k)$ with minimum distance $2k-2$, for every positive integers $n,k$ such that $1<k\leq n/2$, where the latter condition on $k$ can be motivated by a duality argument (see e.g. \cite[Proposition 4.2]{castello2025quasi}). The conjecture was solved for most of the cases in \cite{roth2017construction}, by making use of an algebraic description of these codes known as Sidon spaces. %A Sidon space is an $\fq$-subspace $S$ of $\fqn$ such that if $ab=cd$, then $\lbrace a\fq,c\fq\rbrace=\lbrace b\fq,d\fq\rbrace$. Roth, Raviv and Tamo in \cite{roth2017construction} pointed out that the study of cyclic subspace codes with size $\frac{q^n-1}{q-1}$ and minimum distance $2k-2$ is equivalent to the study of Sidon spaces of $\cG_q(n,k)$ (see \cite[Lemma 34]{roth2017construction}).

Since this connection with coding theory was highlighted, the study of Sidon spaces has increased significantly, and several explicit constructions arose (see e.g. \cite{ben2016subspace,santonastaso2022linearized, roth2017construction,castello2023constructions,castello2025generalized,zhang2022new, zhang2023new,feng2021new,li2023cyclic,wu2025new}). 
Over the years, constructions of multi-orbit cyclic subspace codes have been provided by combining one-orbit cyclic subspace codes in $\cG_q(rk,k)$ with minimum distance $2k-2$, in a way that cardinality increases without reducing the minimum distance (see e.g. \cite{liu2023new,niu2020several,niu2024new,zhang2022constructions,zhang2023constructions,zhang2023further,zhang2025three,zhao2019characterization,zullo2023multi}).

%We will refer to multi-orbit cyclic subspace codes of $\mathcal{G}_q(n,k)$ with minimum distance $2k-2$ as \emph{optimal multi-orbit cyclic subspace codes}.

A central problem in the theory of subspace codes is determining the maximum possible size $\cA_q(n,d,k)$ of a subspace code in $\cG_q(n,k)$ of minimum distance $d$. In recent years, various techniques and methods have been employed to construct larger subspace codes and, consequently, to improve the known lower bounds for $\cA_q(n,d,k)$; we refer to \cite[Chapter 29]{huffman2021concise} and \cite{ kurz2024constructions} for a recent overview. The known upper bounds on $\cA_q(n,d,k)$ are typically the $q$-analogs of the classical bounds established for constant-weight codes. We refer to \cite[Section 3]{cossidente2023combining} and \cite[Section 3.2.1]{heinlein2017tables} for a complete overview and a comparison of the known upper bounds on constant-dimension subspace codes. Among them we mention the \textit{Singleton bound for subspace codes} (\cite[Theorem 9]{koetter2008coding}), \textit{Sphere-packing bound} (\cite[Theorem 6]{koetter2008coding}), \textit{Anticode bound} (\cite[Theorem 1']{ahlswede2001perfect},\cite[Theorem 5.2]{wang2003linear}), \textit{Johnson type bound II} (\cite[Theorem 3]{xia2009johnson},\cite[Theorem 6]{khaleghi2009subspace},\cite[Theorem 4]{etzion2011error})).
The Johnson type bound 
 II for constant dimension subspace codes, (cf. \cite[Theorem 6]{etzion2011error}, \cite[Theorem 7]{khaleghi2009subspace} and \cite[Corollary 3]{xia2009johnson}) reads as:

\begin{equation} \label{eq:johnsonbound}
\cA_q(n,d,k)\leq J_q(n,d,k):=\left\lfloor\frac{q^n-1}{q^k-1}\left\lfloor\frac{q^{n-1}-1}{q^{k-1}-1}\left\lfloor\dots\left\lfloor\frac{q^{n-k+d/2+1}-1}{q^{d/2+1}-1}\left\lfloor\frac{q^{n-k+d/2}-1}{q^{d/2}-1}\right\rfloor\right\rfloor\dots\right\rfloor\right\rfloor\right\rfloor.
\end{equation}

By initially assuming that $k\leq n/2$, it can be proved that the Johnson type bound II improves all upper bounds mentioned so far (see e.g. \cite[Section 4]{kurz2024constructions}) and subsequently, noting that $\mathcal{A}_q(n,d,k)=\mathcal{A}_q(n,d,n-k)$, this holds true in full generality (see e.g. \cite{kschischang2021network}).

A family of codes $\mathcal{F}$ is said to be \emph{asymptotically optimal in \(n\)} if there exists a sequence of codes $(\mathcal{C}_h)_{h \in \mathbb{N}} \subseteq \mathcal{F}$ such that each
$
\mathcal{C}_h \subseteq \mathcal{G}_q(n_h, k_h)
$
is a constant-dimension subspace code with minimum distance \( d_h \), and
\[
\lim_{h \rightarrow \infty} \frac{|\mathcal{C}_h|}{J_q(n_h, d_h, k_h)} = 1.
\]
Analogously, a family of codes $\mathcal{F}$ is said to be \emph{asymptotically optimal in \(q\)} if there exists a sequence of codes $(\mathcal{C}_{q_h})_{h \in \mathbb{N}} \subseteq \mathcal{F}$ such that each
\[
\mathcal{C}_{q_h} \subseteq \mathcal{G}_{q_h}(n, k)
\]
is a constant-dimension subspace code with fixed parameters \( n \), \( k \), and minimum distance \( d \), and
\[
\lim_{h \rightarrow \infty} \frac{|\mathcal{C}_{q_h}|}{J_{q_h}(n, d, k)} = 1.
\]
We will use families of codes whose sizes depend both on $q$ and $n$. Therefore, when we will evaluate the asymptotic behavior of the considered codes, we will consider the limit as $n$ (or $q$) goes to infinity. Also, we will consider the following standard asymptotic notation: given two functions $f(h)$ and $g(h)$, depending on a positive parameter $h$, if $f(h)$ is \textit{asymptotically equal} to $g(h)$, i.e. \\ $\lim\limits_{h \rightarrow +\infty}(f(h)/g(h))=1$, we write 
\[
f(h) \sim g(h), \ \ \ \mbox{as }h \rightarrow + \infty.
\]

Moreover, if $g(h)$ grows faster than $f(h)$, meaning that $\lim\limits_{h \rightarrow +\infty}(f(h)/g(h))=0$ we denote this as 
$f(h)=o_h(g(h))$.

A first construction of asymptotically optimal subspace codes was proposed in~\cite{silva2008rank}, based on lifting rank-metric codes. Since then, several other constructions of asymptotically optimal subspace codes have been developed using a variety of techniques. For a comprehensive overview, we refer the reader to~\cite{kurz2024constructions}. However, in the literature, no cyclic subspace codes are known to be asymptotically optimal, except for the case $n=3k$ and $d=2k-2$, as we will see soon. 
For constant dimension subspace codes in $\cG_q(n,k)$ with minimum distance $d=2k-2$,  \eqref{eq:johnsonbound} reads as follows: 

\begin{equation} \label{eq:johnsonboundoptimal}
\cA_q(n,2k-2,k)\leq J_q(n,2k-2,k):=\left\lfloor\frac{q^n-1}{q^k-1}\left\lfloor\frac{q^{n-1}-1}{q^{k-1}-1}\right\rfloor\right\rfloor.
\end{equation}

%In particular, \eqref{eq:johnsonboundoptimal} tells us that the maximum possible size for an optimal cyclic subspace code of $\cG_q(n,k)$ is \textit{asymptotically equal} to $q^{2n-2k}$ as $k$, or $q$, tends to infinity.

In Table \ref{table:constructionscyclic}, we resume the parameters of the largest known constructions of cyclic subspace codes in $\mathcal{G}_q(n,k)$ and minimum distance $2k-2$. In particular, in fourth column we consider the asymptotic behavior of these sizes $S_i(n,k,q)$, for $i=1,\dots, 5$, as $n$ (or $q$) tends to infinity. Note that in all these constructions we have that $k \mid n$, therefore $n$ goes to infinity as soon as $k$ goes to infinity.
\begin{table}[htp]
	\centering
	\medskip
\small
\begin{tabular}{|c|c|c|c|c|} 
	\hline
&	Parameters & Size & Asymptotic &Reference\\
&	& & behavior &\\
	\hline
    & & & &\\
1) & $n=2k,$ & $S_1(2k,k,q)=\left\lfloor \frac{q-1}{2}\right\rfloor\frac{q^{2k}-1}{q-1}$ &$\sim \frac{1}{2}q^{2k}$ &\cite[Lemma 38]{roth2017construction} \\
& $q>2$& & &\\
\hline 
& & & &\\
2) & $n=4k$ & $S_2(4k,k,q)=\lfloor \frac{q^k-2}{2} \rfloor (q^k-1)(q^{4k}-1)$ & $\sim \frac{1}{2}q^{6k}$ &\cite[Theorem 3.3]{yu2024two} \\
& & & &\\
\hline
& & & &\\
3) & $n=rk,$ & $S_3(rk,k,q)=q^k\frac{q^{(\ell+1)k}-1}{q^k-1}(q^{rk}-1),$ &$\sim q^{(\lfloor\frac{r-1}{2}\rfloor+r)k}$& \cite[Theorem 3.5]{zhang2024large} \\
& $r \geq 3$ & $\ell=\left(\left\lceil \frac{r}{2}\right\rceil-2\right)$ & &\\
\hline 
& & & &\\
4) & $n=rk,$ & $S_4(rk,k,q)=h \left((q^k-1)^h(q^{rk}-1)+\frac{(q^k-1)^{h-1}(q^{rk}-1)}{q-1}\right)$ &$\sim q^{(\frac{3r-1}{2})k}$& \cite[Theorem 2.2]{yu2024two} \\
& $ r=2h+1, \ h \geq 2$ & & &\\
\hline 
& & & &\\
5) & $n=rk,$ & $S_5(rk,k,q)=hq^k (q^k-1)^{h-1}(q^{rk}-1)+\frac{q^{rk}-1}{q^k-1}$ &$\sim q^{(\frac{3r-1}{2})k}$ &\cite[Theorem 2.2]{han2024new} \\
 & $r=2h+1, \ h \geq 2$  & & &\\
\hline
%&&& \\
%6) & $n=5k,$ & $S_6(n,k,q)=(k-1)(5q^k-3)(q^{5k}-1)$ & \cite[Theorem 3.4]{zhang2024constructions} \\
% & & & \\
%\hline
\end{tabular}
\caption{Largest known constructions of cyclic subspace codes in $\mathcal{G}_q(rk,k)$ with minimum distance $2k-2$}\label{table:constructionscyclic}
\end{table}

As mentioned before, we observe that Construction 3) for $r=3$ is the only known construction of cyclic subspace code with minimum distance $2k-2$ whose size asymptotically reaches the Johnson type bound II for $k$ (or $q$) which tends to infinity. Also, the sizes of Construction 1) and 2) in Table \ref{table:constructionscyclic} asymptotically approach the Johnson type bound II within a factor of $1/2+o_k(1)$ (or $1/2+o_q(1)$) and this is the best asymptotic behavior for the size of the known constructions of cyclic subspace codes with the same parameters.

The general idea shared by the constructions mentioned in Table \ref{table:constructionscyclic} consists of taking unions of one-orbit cyclic subspace codes in the same fixed Grassmannian $\cG_q(rk,k)$, with minimum distance $2k-2$, whose representatives satisfy suitable conditions which guarantee that the minimum distance of the obtained code does not decrease. 

\textbf{Our contribution.} 
In this paper, we introduce a new method for constructing large cyclic subspace codes in $\mathcal{G}_q(rk,k)$ with minimum distance $2k-2$, for any composite integer $r$. This time, this new approach consists in combining cyclic subspace codes contained in different Grassmannians. More precisely, starting from a proper tower of field extensions $\F_q < \F_{q^{n_1}} < \F_{q^{n_2}}$, we define an operation $\odot$ that \textit{nests} cyclic subspace codes of $\mathcal{G}_q(n_1,k)$ into cyclic subspace codes of $\mathcal{G}_q(n_2,n_1)$, producing a cyclic subspace code in $\mathcal{G}_q(n_2,k)$. When applied to large cyclic subspace codes with the best possible minimum distance, this operation yields new large cyclic subspace codes with the best minimum distance, as well. As a result, we are able to construct new families of cyclic subspace codes with larger cardinalities than those of previously known constructions and the best possible minimum distance, for all the parameters considered.\\
Surprisingly, this technique leads to the explicit construction of cyclic subspace codes in $\mathcal{G}_q(rk,k)$, with minimum distance $2k-2$, whose sizes asymptotically reach the Johnson type bound II as $k$ or $q$ tends to infinity, for infinitely many values of the integer $r$.
%We emphasize that, prior to this work, only one construction of cyclic subspace code with minimum distance $2k-2$, was known to have size asymptotically attaining the Johnson bound - specifically in $\mathcal{G}_q(3k,k)$, as shown in \cite{zhang2024large}. 
Furthermore, in the case in which the only prime divisors of $r$ are $2$ and $3$, we provide constructions of cyclic subspace codes whose sizes asymptotically approach the Johnson type bound II within a factor of $1/2^e + o_k(1)$, where $e$ is the 2-adic valuation of $r$.
%This asymptotic behavior had previously been observed only for specific constructions \textcolor{red}{of cyclic subspace codes with minimum distance $2k-2$, precisely in $\mathcal{G}_q(2k,k)$, as shown in \cite{roth2017construction} and in $\mathcal{G}_q(4k,k)$, as shown in \cite{yu2024two}.}

%Even more, we observe that thanks to this new technique, in order to construct optimal multi-orbit codes in $\mathcal{G}_q(n,k)$ for $n = rk$ with $r$ composite and with sizes asymptotically reaching the Johnson bound, it is sufficient to construct multi-orbit codes in $\mathcal{G}_q(n',k)$ with $n' = pk$ and $p$ prime, satisfying analogous properties on their parameters.

\textbf{Structure of the paper.}
The paper is organized as follows. %Section \ref{sec:preliminaries} is devoted to preliminaries and auxiliary results. 
Section \ref{sec:nested} introduces the new method to construct multi-orbit codes of $\mathcal{G}_q(rk,k)$ for $r$ composite. 
In Section \ref{sec:specialconstr} we will use this new method to present special constructions of cyclic subspace codes in $\mathcal{G}_q(rk,k)$ with minimum distance $2k-2$, for every composite integer $r$. Section \ref{sec:comparison} is devoted to the comparison of the sizes of the new constructions with the ones of Table \ref{table:constructionscyclic} and the Johnson type bound II. Finally, in Section \ref{sec:conclusions} we leave some open problems and comments.

\begin{comment}
\textbf{Notation.}
We will use the following standard asymptotic notation: given two functions $f(h)$ and $g(h)$, depending on a positive parameter $h$, if $f(h)$ is \textit{asymptotically equal} to $g(h)$, i.e. \\ $\lim\limits_{h \rightarrow +\infty}(f(h)/g(h))=1$, we write 
\[
f(h) \sim g(h), \ \ \ \mbox{as }h \rightarrow + \infty.
\]

Moreover, if $g(h)$ grows faster than $f(h)$, meaning that $\lim\limits_{h \rightarrow +\infty}(f(h)/g(h))=0$ we denote this as 
$f(h)=o_h(g(h))$. %If there exists a constant $C$ depending on $h$ such that $f(h) \leq C g(h)$, we write that $f(h)=O(g(h))$. Furthermore, if $f(h)=O(g(h))$ and $g(h)=O(f(h))$, that is, if the asymptotic behavior of $f(h)$ and $g(h)$ differs only by a constant factor, we write $f(h)=\Theta(g(h))$.
\end{comment}

\section{Nested constructions}
\label{sec:nested}
In this section, we introduce a new technique to construct multi-orbit codes of $\cP_q(n)$ for $n$ composite, by using an operation, that we will denote by $\odot$.  
We will proceed step by step in defining the operation $\odot$, starting with its action on two single subspaces up to how it acts among two multi-orbit codes. We will then investigate the size and the minimum distance of the codes obtained via this operation. Finally, we will show how to extend the operation $\odot$ between more than two multi-orbit codes.

\subsection{Operation $\odot$ between two one-orbit cyclic subspace codes}
Assume that $n = rm$, for some positive integers $r,m>1$. Let $V_1$ be an $\F_q$-subspace of $\F_{q^m}$, $V_2$ an $m$-dimensional $\fq$-subspace of $\fqn$ and $\Phi: \F_{q^m} \to \F_{q^n}$ be an injective $\fq$-linear map such that $V_2 := \im(\Phi)$. Define
\[
V_2 \odot_{\Phi} V_1 := \{\Phi(u) \colon u \in V_1\} \subseteq \F_{q^n}.
\]
\begin{remark}
Note that the operation $\odot_{\Phi}$ depends on the $\fq$-linear map $\Phi$. However, in the next we will write $\odot$ instead of $\odot_{\Phi}$ to lighten the notation.\\
Also, $V_2 \odot V_1= \im(\Phi \vert_{V_1})\subseteq V_2$, where $\Phi \vert_{V_1}$ denotes the restriction of $\Phi$ to $V_1$. Moreover, since $\Phi$ is injective, $\dim_{\F_q}(V_1)=\dim_{\F_q}(\Phi(V_1))=\dim_{\F_q}(V_2 \odot V_1)$.
\end{remark}

We now present an example to illustrate the previous definition. Precisely, we define $V_2\odot V_1\in\cG_q(12,3)$, obtained from $\fq$-subspaces $V_1\in\cG_q(6,3)$ and $V_2\in\cG_q(12,6)$.

\begin{example}
    Let $r=2$, $m=6$ and $n=rm=12$. Let \[
    V_1=\lbrace u+(u+u^q)\eta\colon u\in\F_{q^3}\rbrace\subseteq\F_{q^{6}},
    \]
    where $\eta\in\F_{q^{6}}\setminus\F_{q^3}$. Consider the $\fq$-linear map $\Phi:\F_{q^6}\rightarrow\F_{q^{12}}$ such that $\Phi(v)=v+v^q\gamma$ for any $v\in\F_{q^6}$ where $\gamma\in\F_{q^{12}}\setminus \F_{q^6}$. Let
    \[
    V_2=\im(\Phi)=\lbrace v+v^q\gamma\colon v\in\F_{q^6}\rbrace\subseteq\F_{q^{12}}.
    \]
    Note that the map $\Phi$ is $\fq$-linear and injective. Then we can define
    \[
    \begin{aligned}
        V_2\odot V_1&=\im(\Phi\vert_{V_1})=\lbrace v+v^q\gamma\colon v\in V_1\rbrace\\
        &=\lbrace u+(u+u^q)\eta+u^q\gamma+(u^q+u^{q^2})\eta^q\gamma\colon u\in \F_{q^3}\rbrace\subseteq V_2\subseteq \F_{q^{12}}.
    \end{aligned}
    \]
    $\hfill \lozenge$
\end{example}

%\textcolor{blue}{ripetitivo?
%In this paper, we consider cyclic orbit codes (or one-orbit cyclic subspace codes) the ones of the form
%\[
%cC=\mathrm{Orb}_{\fqn^*}(V)=\lbrace \alpha %V\,\colon\, \alpha\in\fqn^*\rbrace.
%\]
%From now on, we will denote by $\mathrm{Orb}(V)=\mathrm{Orb}_{\fqn^*}(V)$.}

The following lemma allows us to define the operation $\odot$ between the orbit of a subspace $V_2\in\cG_q(n,m)$ and the orbit of a subspace $V_1\in\cG_q(m,k)$.

First, for every $\beta\in\fqn^*$, let $\beta V_2:=\lbrace \beta v \colon v\in V_2\rbrace$ and define 
\[
(\beta V_2)\odot V_1:=\im((\beta\Phi)\vert_{V_1}).
\]

\begin{lemma}
\label{lemma:betaout}
For any $\beta\in\fqn$, it holds
\[
(\beta V_2)\odot V_1=\beta (V_2\odot V_1).
\]
\end{lemma}

\begin{proof}
The assertion immediately follows by observing that
\[
(\beta V_2)\odot V_1=(\beta\Phi)(V_1)=\beta(\Phi(V_1))=\beta(V_2\odot V_1).
\]
\end{proof}

Let $V_1 \in\cG_q(m,k)$ and $\C_{2}=\mathrm{Orb}(V_2)\subseteq \cG_q(n,m)$. Let $\Phi:\F_{q^m}\rightarrow\fqn$ be an $\fq$-linear map such that $V_2=\im(\Phi)$. Then by \Cref{lemma:betaout} we have
\[
\C_{2}\odot V_1:=\{(\beta V_2) \odot V_1 \colon \beta \in \fqn^* \}=\mathrm{Orb}({V_2}\odot V_1)\subseteq \cG_q(n,k).
\]
Similarly, if we consider $\cC_1=\mathrm{Orb}(V_1)$, then
\[
\cC_2\odot \cC_1:= \left\{(\beta V_2) \odot (\alpha V_1) \colon \alpha \in \F_{q^m}^*,\beta \in \F_{q^n}^*\right\} =\bigcup_{\alpha\in\F_{q^m}^*}\mathrm{Orb}(V_2\odot \alpha V_1)\subseteq\mathcal{G}_q(n,k).
\]
We remark that $\mathcal{C}_2=\mathrm{Orb}_{\fqn^*}(V_2)\subseteq\cG_q(n,m)$ and $\mathcal{C}_1=\mathrm{Orb}_{\F_{q^m}^*}(V_1)\subseteq\cG_q(m,k)$. However, in the next, we omit in the notation the Singer cycle acting on $V_1\in\cG_q(n,m)$ and $V_2\in\cG_q(m,k)$ because it will be clear from the context. \\
\\
The previous observation enables us to provide the following definitions.

\begin{definition} \label{def:fullone+fullone}

    Let $n=rm$ and $k\leq m$. Let $V_1 \in\cG_q(m,k)$ and $V_2\in \cG_q(n,m)$. Consider $\Phi:\F_{q^m}\rightarrow\fqn$ be an injective $\fq$-linear map such that $V_2=\im(\Phi)$.
    Denoted by $\C_{2}=\mathrm{Orb}(V_2)\subseteq \cG_q(n,m)$, we define
\[
\C_{2}\odot V_1:=\mathrm{Orb}({V_2}\odot V_1)\subseteq \cG_q(n,k).
\]
Moreover, denoted by $\cC_1=\mathrm{Orb}(V_1)\subseteq\cG_q(m,k)$, we define
\[
\cC_2\odot \cC_1:=\bigcup_{\alpha\in\F_{q^m}^*}\mathrm{Orb}(V_2\odot \alpha V_1)\subseteq\mathcal{G}_q(n,k).
\]
\end{definition}

Note that $\cC_2\odot \cC_1$ is a multi-orbit cyclic subspace code. 

Recall that for $V\in\cG_q(n,k)$, the \emph{stabilizer} of $V$ is defined as
\[
\mathrm{Stab}(V)=\lbrace \alpha\in\fqn^*\mid \alpha V=V\rbrace.
\]
In other words, the $\mathrm{Stab}(V) \cup \{0\}$ is the largest subfield $\F_{q^t}$ of $\fqn$ such that $V$ is an $\F_{q^t}$-subspace of $\fqn$. In \cite{otal2017cyclic}, it is proved the following lemma which connects the cardinality of a cyclic subspace code with the stabilizer of its representative.

\begin{lemma} [see \textnormal{\cite[Theorem 1]{otal2017cyclic}}]
\label{lem:cardinalityifffqt}
Let $V\in\cG_q(n,k)$. The code $\mathrm{Orb}(V)$ has cardinality $\frac{q^n-1}{q^t-1}$ if and only if $\mathrm{Stab}(V)=\F_{q^t}^*$. 
\end{lemma}

By Lemma \ref{lem:cardinalityifffqt}, $\mathrm{Orb}(V)$ has the maximum possible cardinality if and only if $\mathrm{Stab}(V)=\fq^*$. In this case, $\mathrm{Orb}(V)$ is a \emph{full-length one-orbit cyclic subspace code}. %We can consider codes which have more than one orbit, which are called \textbf{multiorbit cyclic subspace codes}, i.e.
%\[ \cC=\displaystyle\bigcup_{i=1}^s\mathrm{Orb}_{\fqn^*}(V_i)=\displaystyle\bigcup_{i=1}^s\lbrace \alpha V_i\,\colon\, \alpha\in\fqn^*\rbrace,
%\]for some $s\geq 1$.
In analogy with the case of one-orbit cyclic subspace codes, we say that a multi-orbit cyclic subspace code $\cC=\bigcup_{i=1}^s\mathrm{Orb}(V_i)$ has \emph{full-length} or that $\cC=\bigcup_{i=1}^s\mathrm{Orb}(V_i)$ is a \textit{full-length multi-orbit cyclic subspace code} if for any $i\in\{1,\ldots,s\}$, $\mathrm{Orb}(V_i)$ is a full-length one-orbit cyclic subspace code, i.e., for any $i\in\{1,\ldots,s\}$, $\beta V_i=V_i$ implies that $\beta\in\fq^*$.

The following result provides sufficient conditions to establish when the codes $\cC_2 \odot V_1$ and $\cC_2 \odot \cC_1$ have full-length.

\begin{proposition} \label{prop:fullone+fullone}
    Let $n=rm$ and $k\leq m$. Let $\cC_1=\mathrm{Orb}(V_1) \subseteq\cG_q(m,k)$ and $\cC_2=\mathrm{Orb}(V_2)\subseteq \cG_q(n,m)$. Consider an injective $\fq$-linear map $\Phi:\F_{q^m}\rightarrow\fqn$ such that $V_2=\im(\Phi)$.
If $\cC_2$ has full-length and $d(\cC_2)>2(m-k)$, then 
$\C_{2}\odot V_1$ 
and $\cC_2\odot\cC_1$ have full-length.
\end{proposition}

\begin{proof}
By Lemma \ref{lem:cardinalityifffqt}, $\cC_2\odot V_1=\mathrm{Orb}(V_2\odot V_1)\subseteq \mathcal{G}_q(n,k)$ has full-length if and only if $\mathrm{Stab}(V_2 \odot V_1)=\fq^*$. Therefore, suppose by contradiction  there exists $\beta \in\fqn\setminus\fq$ such that $\beta(V_2\odot V_1)=V_2\odot V_1$, i.e. $\beta\Phi(V_1)=\Phi(V_1)$. Note that $\beta\Phi(V_1)=\Phi(V_1)\subseteq V_2\cap\beta V_2$. By hypothesis $\Phi$ is injective, and so $\dim_{\fq}(V_2\cap\beta V_2)\geq \dim_{\fq}(V_1)$. Since $\beta\notin\fq$ and $\C_{2}$ has full-length (and hence $V_2 \neq \beta V_2$), we have 
\[
d(\C_{2})\leq d(V_2,\beta V_2)\leq 2\dim_{\fq}(V_2)-2\dim_{\fq}(V_1)=2(m-k)
\]
that contradicts our hypotheses.\\
%Now, we prove that $\C_{2} \odot \C_{1}$ is full-length when $\cC_2$ is full-length and $d(\cC_2)>2(\dim_{\fq}(V_2)-\dim_{\fq}(V_1))$.
Recall that 
\[
\cC_2 \odot \cC_1 = \bigcup_{\alpha\in\F_{q^m}^*}\mathrm{Orb}(V_2\odot \alpha V_1)\subseteq\mathcal{G}_q(n,k). 
\]
Since $\C_{2}$ has full-length and \[d(\cC_2)>2(m-k)=2(\dim_{\fq}(V_2)-\dim_{\fq}(V_1))=2(\dim_{\fq}(V_2)-\dim_{\fq}(\alpha V_1))\] for every $\alpha \in \F_{q^m}^*$, by the previous case, we get that every $\cC_2\odot \alpha V_1=\Orb(V_2 \odot \alpha V_1)$ has full-length. Therefore, $\C_{2} \odot \C_{1}$ has full-length, as well. 
\begin{comment}
    Suppose, by contradiction, there exists $\beta\in\fqn\setminus\fq$ such that
\[
\beta(V\odot \alpha U)=(V\odot \alpha U)
\]
for any $\alpha\in\F_{q^s}^*$, i.e.
\[
\beta\Phi(\alpha U)=\Phi(\alpha U)
\]
for any $\alpha\in\F_{q^s}^*$. Note that $\beta\Phi(\alpha U)=\Phi(\alpha U)\subseteq \beta V\cap V$, which implies
\[
d(\cC_V)\leq d(V,\beta V)\leq 2(\dim_{\fq}(V)-\dim_{\fq}(U)),
\]
a contradiction to the assumptions.
\end{comment}
\end{proof}

\begin{comment}

\begin{corollary} \label{cor:composition1sidon}
If $V$ is a Sidon space of $\F_{q^{n}}$, then $V \odot U$ is an $\F_q$-subspace of $\F_{q^n}$ with $\dim_{\F_q}(V \odot U)=\dim_{\F_q}(U)$ and $V \odot U$ is a Sidon space of $\F_{q^n}$ as well.
\end{corollary}

\begin{proof}
    The subspace $V \odot U$ is the image of the map $\Phi$ restricted to the subspace $U$. Therefore, since $\Phi$ is injective, we have
    \[\dim_{\F_q}(V \odot U)=\dim_{\F_q}(\im(\Phi\vert_U))=\dim_{\F_q}(U).\]
    Finally, since $V \odot U$ is a subspace of $V$, and $V$ is a Sidon space, we get that $V \odot U$ is a Sidon space as well. 
\end{proof}
\end{comment}

We present an illustrative example to show how \Cref{prop:fullone+fullone} can be applied.

\begin{example}
\label{ex:RRToneorbits}
     Let $r=3,m=3k$ and $n = 9k$ with $k>1$. Let $\gamma_1 \in \F_{q^{3k}}$ be such that $\F_{q^k}(\gamma_1)=\F_{q^{3k}}$. Consider the subspace
\[
V_{1}=\{v+v^q\gamma_1 \colon v \in \F_{q^k} \} \subseteq \mathcal{G}_q(3k,k),
\]
and consider the code 
$
\cC_1=\Orb(V_{1}) \subseteq  \mathcal{G}_q(3k,k).$
Let $\gamma_2 \in \F_{q^{9k}}$ be such that $\F_{q^{3k}}(\gamma_2)=\F_{q^{9k}}$. Consider the subspace
\[
V_{2}=\{u+u^q\gamma_2 \colon u \in \F_{q^{3k}} \} \subseteq \mathcal{G}_q(9k,3k)
\]
and define the code $
\C_{2}=\Orb(V_{2}) \subseteq \mathcal{G}_q(9k,3k)$. By \cite[Theorem 12 and Lemma 34]{roth2017construction}, $\cC_2$ is a full-length orbit code with    $d(\cC_2)=2(3k)-2$.\\
Also, the map 
\[
\begin{array}{llll}
\Phi_{2}: & \F_{q^{3 k}} & \longrightarrow & \F_{q^{9k}}\\
& u & \longmapsto & u+u^q\gamma_2
\end{array}
\]
is an injective $\F_q$-linear map
and $
V_{2}=\im(\Phi_{2}).$
Therefore, we can consider 
\[
\C_{2} \odot \C_{1}=\bigcup_{\alpha \in \F_{q^{3k}}^*} \Orb (V_{2} \odot \alpha V_{1}).
\]
Precisely,
\[
\cC_2\odot\cC_1=\bigcup_{\alpha\in\F_{q^{3k}}^*}\bigcup_{\beta\in\F_{q^{9k}}^*}\lbrace \beta\alpha (v+v^q\gamma_1)+\beta\alpha^q (v+v^q\gamma_1)^q\gamma_2\colon v\in\F_{q^{k}}\rbrace\subseteq \cG_q(9k,k).
\]
Since $\cC_2$ has full-length and $d(\cC_2)=2(3k)-2> 2(\dim_{\fq}(V_2)-\dim_{\fq}(V_1))=4k$, by Proposition \ref{prop:fullone+fullone}, $\cC_2\odot \cC_1$ is a full-length multi-orbit cyclic subspace code.

$\hfill \lozenge$
\end{example}

\subsection{Operation $\odot$ between two multi-orbit cyclic subspace codes}
We point out that the operation $\odot$ allows us to construct a multi-orbit cyclic subspace code $\C_2 \odot \C_1$ starting from two one-orbit cyclic subspace codes $\C_1$ and $\C_2$. Now, in the same spirit of \Cref{def:fullone+fullone}, thanks to Lemma \ref{lemma:betaout}, we can define the operation $\odot$ between multi-orbit cyclic subspace codes, which will produce large multi-orbit codes.

\begin{definition} \label{def:multiorb+multiorb}
    Let $n=rm$, with $r,m >1$. Let $1<k\leq m$ and $s_1,s_2\geq 1$ be positive integers.  Consider
\[
\C_1=\bigcup_{h_1=1}^{s_1} \Orb(V_{1,h_1})\subseteq \mathcal{G}_q(m,k),
\]
and 
\[
\C_2=\bigcup_{h_2=1}^{s_2} \Orb(V_{2,h_2})\subseteq \mathcal{G}_q(n,m), 
\]
multi-orbit cyclic subspace codes, where $V_{2,h_2}=\im(\Phi_{2,h_2})$ for some injective $\F_q$-linear map
\[
\Phi_{2,h_2}: \F_{q^{m}} \longrightarrow \F_{q^{n}}, 
\]
for every $h_2 \in \{1,\ldots,s_2\}$.
Define $\C_2 \odot \C_{1}$ as 
\[
\begin{array}{rl}
\cC_2 \odot \cC_1  := & \bigcup\limits_{h_2=1}^{s_2} \bigcup\limits_{h_1=1}^{s_1} \Orb(V_{2,h_2}) \odot \Orb(V_{1,h_1})
\\
= & \bigcup\limits_{h_2=1}^{s_2}\bigcup\limits_{h_1=1}^{s_1}\bigcup\limits_{\alpha\in\F_{q^{m}}^*}\Orb(V_{2,h_2} \odot  \alpha V_{1,h_1})\\  \subseteq & \mathcal{G}_q(n,k). 
\end{array}
\]
\end{definition}

Note that $\cC_2\odot\cC_1$ is a multi-orbit cyclic subspace code in $\mathcal{G}_q(n,k)$. In a similar manner to \Cref{prop:fullone+fullone}, in the following proposition, we provide sufficient conditions to determine when the operation $\odot$ on multi-orbit cyclic subspace codes produces full-length codes.

\begin{proposition} \label{prop:fulllenghtmulti}
Let $n=rm$ with $r,m >1$ and let $1<k\leq m$. Consider $\cC_1\subseteq \cG_q(m,k)$ and $\cC_2\subseteq\cG_q(n,m)$ as in \Cref{def:multiorb+multiorb}.
If $\cC_2$ has full-length and $d(\cC_2)>2(m-k)$, then $\cC_2\odot \cC_1$ is a full-length multi-orbit cyclic subspace code.
\end{proposition}

%Note that $\lvert \C_i \rvert=t_i$ and $t_1=s_1\frac{q^{r_1k}-1}{q-1}$, $t_2=s_2\frac{q^{r_1r_2k}-1}{q-1}$. 
\begin{proof}
\begin{comment}
Note that, since by Lemma \ref{lemma:betaout} $\beta(V_{2,h_2} \odot  \alpha V_{1,h_1})=(\beta V_{2,h_2}) \odot  \alpha V_{1,h_1}$ for any $\beta\in\fqn^*$ and $\alpha\in\F_{q^m}^*$, we have
\[
\begin{aligned}
\cC_2 \odot \cC_1&= \bigcup_{h_2=1}^{s_2}\bigcup_{\alpha\in\F_{q^{m}}^*}\Orb(V_{2,h_2} \odot  \alpha V_{1,h_1}) \\
&=\bigcup_{h_2=1}^{s_2}\bigcup_{\alpha\in\F_{q^{m}}^*} \cC_{V_{2,h_2}} \odot (\alpha V_{1,h_1})\\
&=\bigcup_{h_2=1}^{s_2} \cC_{V_{2,h_2}} \odot \cC_{V_{1,h_1}}.
\end{aligned}
\]
\end{comment}
Recall that 
\[
\cC_2 \odot \cC_1  =  \bigcup_{h_2=1}^{s_2} \bigcup_{h_1=1}^{s_1} \Orb(V_{2,h_2}) \odot \Orb(V_{1,h_1}).
\]
To show that $\C_2 \odot \C_1$ has full-length, by definition we need to prove that $\Orb(V_{2,h_2}) \odot \Orb(V_{1,h_1})$ has full-length, for any $h_1 \in \{1,\ldots,s_1\}$ and $h_2 \in \{1,\ldots,s_2\}$. By Proposition \ref{prop:fullone+fullone}, it is sufficient to show that  $\Orb(V_{2,h_2})$ has full length and $d(\Orb(V_{2,h_2}))>2(m-k)$. Note that, since $\C_2=\bigcup_{h_2=1}^{s_2} \Orb(V_{2,h_2})$ has full-length, we have that by definition $\Orb(V_{2,h_2})$ has full-length for every $h_2\in \{1,\ldots,s_2\}$. Moreover,
\[
d(\Orb(V_{2,h_2})) \geq d(\cC_2)>2(m-k),
\]
for every $h_2\in \{1,\ldots,s_2\}$. Thus by \Cref{prop:fullone+fullone} the assertion follows.
\begin{comment}
Note that 
\[
\begin{aligned}
\C_2 \odot \C_1=& \bigcup_{h_2=1}^{s_2}\bigcup_{\alpha\in\F_{q^{m}}^*}\Orb(V_{2,h_2} \odot  \alpha V_{1,h_1}) \\
& =\bigcup_{h_2=1}^{s_2}\bigcup_{\alpha\in\F_{q^{m}}^*}\lbrace \beta(V_{2,h_2} \odot  \alpha V_{1,h_1})\colon \beta \in \fqn^*\rbrace\\
&=\bigcup_{h_2=1}^{s_2}\bigcup_{\alpha\in\F_{q^{m}}^*}\lbrace \beta(V_{2,h_2} \odot  \alpha V_{1,h_1})\colon \beta \in \fqn^*\rbrace\\
&=\bigcup_{\beta\in\fqn^*}\bigcup_{\alpha\in\F_{q^{m}}^*}\lbrace \beta V_{2,h_2} \odot  \alpha V_{1,h_1}\rbrace.\\
\end{aligned}
\]
Therefore,
\[
\cC_2 \odot \cC_1 =\bigcup_{\beta\in\fqn^*}\bigcup_{\alpha\in\F_{q^{m}}^*}\lbrace \beta V_{2,h_2} \odot  \alpha V_{1,h_1}\rbrace=\bigcup_{j_1=1}^{t_1}\bigcup_{j_2=1}^{t_2}\lbrace U_{1,j_1}\odot U_{2,j_2}\rbrace.
\]
Also, if $\Phi_{2,h_2}(U_{1,j_1})$ has full-length for any $j_1=1,\dots,t_1$ and any $h_2=1,\dots,s_2$ then
\[
\cC_2 \odot \cC_1 =\bigcup_{\beta\in\fqn^*}\bigcup_{\alpha\in\F_{q^{m}}^*}\lbrace \beta V_{2,h_2} \odot  \alpha V_{1,h_1}\rbrace=\bigcup_{\beta\in\fqn^*}\bigcup_{\alpha\in\F_{q^{m}}^*}\{\beta \Phi_{2,h_2}(\alpha V_{1,h_1})\}
\]
is trivially full-length. So suppose, by contradiction, there exists $\gamma\in\fqn\setminus\fq$ such that
\[
\gamma(V_{2,h_2}\odot \alpha V_{1,h_1})=(V_{2,h_2}\odot \alpha V_{1,h_1})
\]
thus
\[
\gamma(V_{2,h_2}\odot \alpha V_{1,h_1})=(V_{2,h_2}\odot \alpha V_{1,h_1})\subseteq V_{2,h_2}\cap\gamma V_{2,h_2}.
\]
This implies that
\[
d(\cC_2)\leq d(V_{2,h_2},\gamma V_{2,h_2})= 2(\dim_{\fq}(V_{2,h_2})-\dim_{\fq}(V_{1,h_1}))
\]
a contradiction.
\end{comment}
\end{proof}

We give an example of how \Cref{prop:fulllenghtmulti} can be applied.

\begin{example}
\label{ex:Largeoneorbits}
Let $n = 9k$, with $k>1$. We will construct multi-orbit cyclic subspace codes in $\mathcal{G}_q(9k,k)$ by using codes in $\mathcal{G}_q(3k,k)$ and $\mathcal{G}_q(9k,3k)$. Let us start by defining the codes that we will employ. Let $\gamma_1 \in \F_{q^{3k}}$ be such that $\F_{q^k}(\gamma_1)=\F_{q^{3k}}$. For any $a_1 \in \F_{q^k}$, consider the subspace
\[
V_{1,a_1}=\{v+(v^q+a_1v)\gamma_1 \colon v \in \F_{q^k} \} \subseteq \mathcal{G}_q(3k,k),
\]
and define the code 
\[
\cC_1=\bigcup_{a_1 \in \F_{q^k}} \Orb(V_{1,a_1})=\bigcup_{\alpha \in \F_{q^{3k}}^*} \bigcup_{a_1 \in \F_{q^k}} (\alpha V_{1,a_1}) \subseteq  \mathcal{G}_q(3k,k).
\]
 We can now use a code of the same shape of $\C_1$ but in $\mathcal{G}_q(3\overline{k},\overline{k})$, with $\overline{k}=3k$. Precisely, let $\gamma_2 \in \F_{q^{9k}}$ be such that $\F_{q^{3k}}(\gamma)=\F_{q^{9k}}$. For any $a_2 \in \F_{q^{3k}}$, consider the subspace
\[
V_{2,a_2}=\{u+(u^q+a_2u)\gamma_2 \colon u \in \F_{q^{3k}} \} \subseteq \mathcal{G}_q(9k,3k),
\]
and define the code 
\[
\cC_2=\bigcup_{a_2 \in \F_{q^{3k}}} \Orb(V_{2,a_2})=\bigcup_{\alpha \in \F_{q^{9k}}^*} \bigcup_{a_2 \in \F_{q^{3k}}} (\alpha V_{2,a_2}) \subseteq \mathcal{G}_q(9k,3k).
\]
Note that, by \cite[Lemma 3.4]{zhang2024large}(without considering the orbits defined by different $\omega^i$'s), $\cC_2$ has minimum distance $d(\cC_2)=2(3k)-2$. Also, $\Orb(V_{2,a_2})$ has minimum distance $2(3k)-2$ and $\Orb(V_{2,a_2})$ has full-length (see \cite[Lemma 3.1]{zhang2024large} and \cite[Lemma 34]{roth2017construction}), and so by definition $\cC_2$ has full-length, as well.
For every $a_2 \in \F_{q^{3k}}$, the map 
\[
\begin{array}{llll}
\Phi_{2,a_2}: & \F_{q^{3 k}} & \longrightarrow & \F_{q^{9k}}\\
& u & \longmapsto & u+(u^q+a_2u)\gamma_2
\end{array}
\]
is an injective $\F_q$-linear map
such that 
\[
V_{2,a_2}=\im(\Phi_{2,a_2}).
\]
Therefore, we can define 
\[
\C_2 \odot \C_{1}=\bigcup_{a_2 \in \F_{q^{3k}}} \bigcup_{\alpha \in \F_{q^{3k}}^*} \bigcup_{a_1 \in \F_{q^k}} \Orb (V_{2,a_2} \odot (\alpha V_{1,a_1}))\subseteq \cG_q(9k,k),
\]
%Precisely,
%\[
%\cC_2\odot\cC_1
%=\bigcup_{a_1\in\F_{q^{k}}}\bigcup_{a_2\in\F_{q^{3k}}}\bigcup_{\alpha\in\F_{q^{3k}}^*} \Orb (V_{a_1,a_2,\alpha}) ,
%\]
where
\[
    V_{2,a_2} \odot (\alpha V_{1,a_1})=\{ (\alpha v+\alpha (v^q+a_1v)\gamma_1) +(\alpha^{q}v^q+\alpha^q(v^q+a_1v)^{q}\gamma_1^q+a_2(\alpha v+\alpha (v^q+a_1v)\gamma_1))\gamma_2\colon v\in\F_{q^{k}}\}.
\]
Since $\cC_2$ has full-length and  $d(\cC_2)=2(3k)-2> 2(3k-k)=4k$, by Proposition \ref{prop:fulllenghtmulti}, $\cC_2\odot \cC_1$ is a full-length multi-orbit cyclic subspace code.

   $\hfill \lozenge$
\end{example}

\subsection{Parameters of $\cC_2\odot\cC_1$}
\noindent In what follows, we study the parameters of codes obtained by the operation $\odot$. Specifically, we will see how to get information on the cardinality and on the minimum distance of the codes obtained by using $\odot$.

Before going to this we recall that in \cite{roth2017construction}, the authors highlighted an equivalence between Sidon spaces of dimension $k$ and one-orbit cyclic subspace codes with minimum distance $2k-2$. A \emph{Sidon space} is an $\fq$-subspace $S$ of $\fqn$ such that if $ab=cd$, then $\lbrace a\fq,c\fq\rbrace=\lbrace b\fq,d\fq\rbrace$. Roth, Raviv and Tamo in \cite{roth2017construction} pointed out that every $k$-dimensional Sidon space in $\fqn$ is equivalent to a one-orbit cyclic subspace code of cardinality $\frac{q^n-1}{q-1}$ and minimum distance $2k-2$.

\begin{theorem} [see \textnormal{\cite[Lemma 34]{roth2017construction}}]
\label{thm:sidoniffoptimal}
Let $U \in\cG_q(n,k)$. The code $\mathrm{Orb}(U)$ has minimum distance $2k-2$ and size $\frac{q^n-1}{q-1}$  if and only if $U$ is a Sidon space, i.e.
\[
\dim_{\fq}(U\cap\alpha U)\leq 1\quad \text{ for any }\alpha\in\fqn\setminus\fq.
\]
\end{theorem}
 %With an analogue argument, we see that the same happens also for multi-orbit cyclic subspace codes.

\begin{remark}
\label{rmk:ifd=2k-2thenOrbisfull} 
 We observe that if $V\in\cG_q(n,k)$, with $k>1$, and  $d(\mathrm{Orb}(V))=2k-2$, then $\mathrm{Orb}(V)$ has necessarily full-length. Indeed, if $d(\mathrm{Orb}(V))=2k-2$, then by Theorem \ref{thm:sidoniffoptimal}, $V$ is a Sidon space, i.e. $\dim_{\fq}(V\cap\alpha V)\leq 1$ for any $\alpha\in\fqn\setminus\fq$. This means that if $\alpha\in\mathrm{Stab}(V)$ and by contradiction we suppose $\alpha\notin\fq$, on one hand we have $V=\alpha V$ and so $\dim_{\fq}(V\cap\alpha V)=k$. On the other hand $\dim_{\fq}(V\cap\alpha V)\leq 1$ for any $\alpha\in\fqn\setminus\fq$, a contradiction since $k>1$. %Similarly, if $\cC=\bigcup_{i=1}^s\mathrm{Orb}(V_i)=\{U_1,\ldots,U_t\}\subseteq\cG_q(n,k)$ is a multi-orbit cyclic subspace code with $d(\cC)=2k-2$, then every $\mathrm{Orb}(V_i)$ has full-length and hence $\cC$ is necessarily full-length. %Indeed, if $d(\cC)=2k-2$, then by Lemma \ref{lm:equivalenceoptimalsubspaces},   $\dim_{\fq}(U_i\cap U_j)\leq 1$ for any $i,j\in\lbrace 1,\dots,t\rbrace$ and $i\neq j$. This means that if $U_i,U_j\in\mathrm{Orb}(V_h)$ for some $h\in\{1,\dots,t\}$ and $U_i\neq U_j$. Then $U_i=\alpha V_h$ and $U_j=\beta V_h$ for some $\alpha, \beta\in\fqn^*$ such that $\frac{\beta}{\alpha}\notin\fq$. Then by Remark \ref{rmk:ifd=2k-2thenOrbisfull}, $\mathrm{Stab}(V_h)=\fq^*$, for any $h\in\lbrace 1, \dots,s \rbrace$.
\end{remark}

Finally, we introduce the following notation for $\cC_2\odot \cC_1$ before going to the proof of the main theorem of this subsection.

Let $n=rm$ with $r,m>1$ and let $k>1$ and $s_1,s_2\geq 1$ be positive integers. Consider $s_1$ distinct one-orbit cyclic subspace codes $\Orb(V_{1,h_1})\subseteq\cG_q(m,k)$ for $h_1\in\lbrace 1,\dots,s_1\rbrace$ and $s_2$ distinct one-orbit cyclic subspace codes $\Orb(V_{2,h_2})\subseteq\cG_q(n,m)$ for $h_2\in\lbrace 1,\dots,s_2\rbrace$. Let
\begin{equation}
\label{eq:C1}
\cC_1=\bigcup_{h_1=1}^{s_1} \Orb(V_{1,h_1})=\{U_{1,1},\ldots,U_{1,t_1}\} \subseteq \mathcal{G}_q(m,k),
\end{equation}
and 
\begin{equation}
\label{eq:C2}
\cC_2=\bigcup_{h_2=1}^{s_2} \Orb(V_{2,h_2})=\{U_{2,1}
,\ldots,U_{2,t_2}\} \subseteq \mathcal{G}_q(n,m).
\end{equation}
%with $\Orb(V_{1,i}) \neq \Orb(V_{1,j})$ and $\Orb(V_{2,i'}) \neq \Orb(V_{2,j'})$, for every $i \neq j, i,j\in\lbrace 1,\dots,s_1\rbrace$ and every $i' \neq j', i',j'\in\lbrace 1,\dots,s_2\rbrace$. 
Since the orbits are distinct, their intersection is trivial and so $\lvert\cC_1\rvert=t_1=\displaystyle\sum_{h_1=1}^{s_1}\lvert\mathrm{Orb}(V_{1,h_1})|$ and $\lvert\cC_2\rvert=t_2=\displaystyle\sum_{h_2=1}^{s_2}\lvert\mathrm{Orb}(V_{2,h_2})|$.\\
 For every $h_2 \in \{1,\ldots,s_2\}$, let \[
\Phi_{2,h_2}: \F_{q^{m}} \longrightarrow \F_{q^{n}} 
\] be an injective $\F_q$-linear map such that $V_{2,h_2}=\im(\Phi_{2,h_2})$. Then defining \[
\cC=\cC_2 \odot \cC_1= \bigcup_{h_2=1}^{s_2}\bigcup_{h_1=1}^{s_1}\bigcup_{\alpha\in\F_{q^{m}}^*}\Orb(V_{2,h_2} \odot  \alpha V_{1,h_1})=\bigcup_{h_2=1}^{s_2}\bigcup_{j_1=1}^{t_1}\Orb(V_{2,h_2} \odot U_{1,j_1})
\subseteq\cG_q(n,k),\]
we observe that by Lemma \ref{lemma:betaout} we may also write $\cC$ as
\begin{equation}
\label{eq:C2dotC1}
\cC=\cC_2 \odot \cC_1=\bigcup_{j_2=1}^{t_2}\bigcup_{j_1=1}^{t_1}(U_{2,j_2} \odot U_{1,j_1})\subseteq \mathcal{G}_q(n,k).
\end{equation}

In the following remark, we observe that if $d(\cC)=2k-2\ell$ then the intersection of any two distinct codewords is bounded by $\ell$.

\begin{remark}
\label{rmk:intersectionbounded}
Note that if the minimum distance of $\cC=\bigcup_{h=1}^s\mathrm{Orb}(V_h)=\lbrace U_1,\dots, U_t\rbrace$ is $2k-2\ell$, then for any $U_i,U_j\in\cC$ with $i\neq j$ it results
\[
2k-2\ell=d(\cC)\leq d(U_i,U_j)=2k-2\dim_{\fq}(U_i\cap U_j)
\]
and so
\[
\dim_{\fq}(U_i\cap U_j)\leq \ell.
\]
\end{remark}

The following theorem tells us which is the minimum distance and the cardinality of $\cC_2\odot \cC_1$, by starting from the cardinalities and the minimum distances of $\cC_2$ and $\cC_1$.

\begin{theorem} \label{th:nested2}
Let $\cC_1\subseteq\cG_q(m,k)$ and $\cC_2\subseteq\cG_q(n,m)$ as in \eqref{eq:C1} and \eqref{eq:C2}, respectively, and assume that $d(\C_1)=2k-2\ell$ and $d(\C_2)=2m-2\ell'$. The code $\cC_2\odot\cC_1$ is a multi-orbit cyclic subspace code with minimum distance at least $2k-2\max\{\ell, \ell'\}$. Also, if $\max\{\ell,\ell'\}<k$, then $\C_2 \odot \C_1$ has size $\lvert \cC_2\odot \cC_1\rvert=\lvert\cC_2\rvert\lvert\cC_1\rvert$ and if $\max\{\ell, \ell'\}=\ell$, then $d(\cC_2\odot \cC_1)=2k-2\ell.$
\end{theorem}

\begin{proof}
We note that since the minimum distances of $\C_1$ and $\C_2$ are $2k-2\ell$ and $2m-2\ell'$, respectively, by Remark \ref{rmk:intersectionbounded} we have
\begin{equation} \label{eq:cond1codoptimal}
\dim_{\F_q}(U_{1,j_1} \cap U_{1,j_1'})  \leq \ell,
\end{equation}
for every $j_1,j_1' \in \{1,\ldots,t_1\}$, with $j_1 \neq j_1'$, and that 
\begin{equation} \label{eq:cond2codoptimal}
\dim_{\F_q}(U_{2,j_2} \cap U_{2,j_2'})  \leq \ell',
\end{equation}
for every $j_2,j_2' \in \{1,\ldots,t_2\}$, with $j_2 \neq j_2'$ .\\
We prove that when $(j_1,j_2)\neq (j_1',j_2')$, then 
    \begin{equation} \label{eq:optimal2nested}
     \dim_{\F_q}\left((U_{2,j_2} \odot U_{1,j_1}) \cap ( U_{2,j_2'}\odot U_{1,j_1'})\right) \leq \begin{cases}
         \ell' & \text{ if } j_2 \neq j_2',\\
         \ell & \text{ if } j_2 = j_2'.
     \end{cases}
    \end{equation}
    We divide the discussion into two cases.

    \noindent \textbf{(Case $j_2 \neq j_2'$.)} First, note that $U_{2,j_2} \odot U_{1,j_1}$ is an $\F_q$-subspace of $U_{2,j_2}$, for every $j_2 \in \{1,\ldots,t_2\}$ and for every $j_1\in \{1,\ldots,t_1\}$. 
    Therefore, by using \eqref{eq:cond2codoptimal}, we have 
    \[
    \dim_{\F_q}((U_{2,j_2} \odot U_{1,j_1}) \cap  (U_{2,j_2'} \odot U_{1,j_1'})) \leq \dim_{\F_q}(U_{2,j_2} \cap  U_{2,j_2'})  \leq \ell'. 
    \]
\noindent \textbf{(Case $j_2 = j_2'$.)} Since we are assuming that $(j_1,j_2)\neq (j_1',j_2')$, then $j_1 \neq j_1'$. Now, $U_{2,j_2} \in \C_2=\bigcup_{h_2=1}^{s_2} \Orb(V_{2,h_2})$, so there exists $\alpha \in \F_{q^n}^*$ such that $U_{2,j_2}=\alpha V_{2,h_2}$, for some $h_2 \in \{1,\ldots,s_2\}$ and $V_{2,h_2}=\im(\Phi_{2,h_2})$. Therefore $U_{2,j_2}=\alpha \im(\Phi_{2,h_2})=\im(\alpha \Phi_{2,h_2})$. As a consequence, we have \[
U_{2,j_2} \odot U_{1,j_1}=\alpha\Phi_{2,h_2}(U_{1,j_1})\]
and 
\[ U_{2,j_2} \odot U_{1,j_1'}=\alpha\Phi_{2,h_2}(U_{1,j_1'}).\]
Since $\Phi_{2,h_2}$ is injective for every $h_2=1,\dots, j_2$, we get that for every $j_2\in\lbrace 1, \dots, t_2\rbrace$ and $j_1\in\lbrace 1, \dots, t_1\rbrace$, 
\[
\begin{array}{rl}
    \dim_{\F_q}((U_{2,j_2} \odot U_{1,j_1}) \cap  (U_{2,j_2} \odot U_{1,j_1'})) & =\dim_{\F_q}(\alpha \Phi_{2,h_2}(U_{1,j_1} \cap  U_{1,j_1'})) \\
    & =\dim_{\F_q}(U_{1,j_1} \cap  U_{1,j_1'}) \\
    &\leq \ell, 
    \end{array}
    \]
where the last inequality follows from \eqref{eq:cond1codoptimal}. This proves \eqref{eq:optimal2nested}. As a consequence, when $(j_1,j_2)\neq (j_1',j_2')$, we get $\dim_{\F_q}((U_{2,j_2} \odot U_{1,j_1}) \cap  (U_{2,j_2} \odot U_{1,j_1'})) \leq \max\{\ell,\ell'\}$ and, since $\C_2 \odot \C_1 \subseteq \cG_q(n,k)$, by Remark \ref{rmk:intersectionbounded}, we obtain the first part of the assertion.  \\ Now, by \eqref{eq:C2dotC1},
\[
\lvert \cC_2 \odot \cC_1 \rvert= \left\lvert \bigcup_{j_2=1}^{t_2}\bigcup_{j_1=1}^{t_1}(U_{2,j_2} \odot U_{1,j_1}) \right\rvert. 
\]
Hence, if $\max\{\ell,\ell'\}<k$, then we have by \eqref{eq:optimal2nested} that the subspaces $U_{2,j_2} \odot U_{1,j_1} \in \mathcal{G}_q(n,k)$ are all distinct and so
\[\lvert \C_2 \odot \C_1 \rvert =t_1t_2=\lvert \C_2 \rvert\lvert \C_1 \rvert.\]
Finally, assume that $\max\{\ell,\ell'\}=\ell$. Therefore, $d(\C_2 \odot \C_1) \geq 2k-2\ell$. On the other hand, $d(\C_1)=2k-2\ell$, implying that there exist $U_{1,j_1}, U_{1,j_1'} \in \C_1$ such that $\dim_{\F_q}(U_{1,j_1} \cap U_{1,j_1'}) = \ell$. Now, let $U_{2,j_2} = \im(\alpha \Phi_{2,h_2}) \in \C_2$, for $h_2\in\lbrace 1, \dots, s_2\rbrace$, with $\alpha \in \F_{q^n}^*$. As a consequence,
\[
\begin{array}{rl}
    \dim_{\F_q}((U_{2,j_2} \odot U_{1,j_1}) \cap  (U_{2,j_2} \odot U_{1,j_1'})) & =\dim_{\F_q}(\alpha \Phi_{2,h_2}(U_{1,j_1} \cap  U_{1,j_1'})) \\
    & =\dim_{\F_q}(U_{1,j_1} \cap  U_{1,j_1'}) \\
    &= \ell, 
    \end{array},
    \]
    proving that $d(\C_2 \odot \C_1)=2k-2\ell$.
\end{proof}

As a consequence of the previous result, we can prove that when the cyclic subspace codes $\C_1\subseteq \cG_q(m,k)$ and $\C_2\subseteq \cG_q(n,m)$ have minimum distances $2k-2$ and $2m-2$, respectively, then the code $\C_2 \odot \C_1\subseteq\cG_q(n,k)$ is a cyclic subspace code with minimum distance $2k-2$. Moreover, in this case if $\C_2$ has full-length, then we get that $\C_2 \odot \C_1$ has full-length as well.

\begin{corollary}
\label{cor:C1C2optimalimpliesC2odotC1optimal}
Let $n=rm$ and consider $\cC_1\subseteq\cG_q(m,k)$, with $1<k\leq m$, $\cC_2\subseteq\cG_q(n,m)$ and $\cC_2\odot\cC_1\subseteq\cG_q(n,k)$ as in Theorem \ref{th:nested2}. If $\ell=1$ and $\ell'=1$, i.e. $\cC_1$ and $\cC_2$ have minimum distances $2k-2$ and $2m-2$, respectively, then $\cC_2\odot\cC_1\subseteq \cG_q(n,k)$ is a multi-orbit cyclic subspace code having minimum distance $2k-2$ and size
\begin{equation}
\label{eq:cardinalityoptimal}
\lvert \cC_2 \odot \cC_1 \rvert=\lvert \cC_2 \rvert \lvert \cC_1 \rvert . 
\end{equation}
Also, if $\C_2$ has full-length then $\cC_2 \odot \cC_1$ has full-length, as well.
\end{corollary}

\begin{proof}
    %The code $\C_2 \subseteq \mathcal{G}_q(n,m)$ has minimum distance $2m-2$. By \Cref{rmk:ifd=2k-2thenMultiOrbisfull}, we have that $\C_2$ has full-length. Moreover,
   % \[
    %2m-2=d(\C_2)>2(m-k),
   % \]
   % since $k> 2$, and so  \Cref{prop:fulllenghtmulti} implies that $\C_2 \odot \C_1$ has full-length. Finally, \Cref{th:nested2}, for the case $\ell=\ell'=1$, implies that $\C_2 \odot \C_1$ has minimum distance $2k-2$ and cardinality as in \eqref{eq:cardinalityoptimal}.
   By applying \Cref{th:nested2} for the case $\ell=\ell'=1$, we derive that $\C_2 \odot \C_1$ has minimum distance $2k-2$ and cardinality as in \eqref{eq:cardinalityoptimal}.  Furthermore, if $\C_2$ has full-length and $d(\cC_2)=2m-2>2(m-k)$, then $\C_2 \odot \C_1$ has full-length by \Cref{prop:fulllenghtmulti}, thus proving the assertion.
\end{proof}

We illustrate, in the following example, a construction of a family of cyclic subspace codes in $\cG_q(9k,k)$ with minimum distance $2k-2$ and size asymptotically equal to $q^{12k-2}$, as $k$ or $q$ tends to infinity. We will use the operation $\odot$ acting on codes defined as in Example \ref{ex:RRToneorbits}.

\begin{example} \label{ex:nest2one-orbit}
Consider the codes $\cC_1$ and $\cC_2$ as in Example \ref{ex:RRToneorbits}. As proved in \cite[Theorem 12]{roth2017construction}, $\C_{1}$ is a cyclic subspace code having size $t_1=\lvert \C_{1} \rvert= \frac{q^{3k}-1}{q-1}$ and minimum distance $2k-2$, and $\C_{2}$ is a cyclic subspace code having size $t_2=\lvert \C_{2} \rvert= \frac{q^{9k}-1}{q-1}$ and minimum distance $2(3k)-2$. 
So by \Cref{cor:C1C2optimalimpliesC2odotC1optimal}
\[
\cC_2\odot\cC_1=\bigcup_{\alpha\in\F_{q^{3k}}^*}\bigcup_{\beta\in\F_{q^{9k}}^*}\lbrace \beta\alpha (v+v^q\gamma_1)+\beta\alpha^q (v+v^q\gamma_1)^q\gamma_2\colon v\in\F_{q^{k}}\rbrace\subseteq \cG_q(9k,k)
\]
is a multi-orbit cyclic subspace code in $\mathcal{G}_{q}(n,k)$ having size $\lvert \C_{1} \rvert \lvert \C_{2} \rvert=\frac{(q^{3k}-1)(q^{9k}-1)}{(q-1)^2}$ and minimum distance $2k-2$.

$\hfill \lozenge$
\end{example}

In Example \ref{ex:ex:nest2multi-orbit} we provide an example of cyclic subspace code with the same parameters of \Cref{ex:nest2one-orbit}, with exception of the size of the code. Indeed, the construction provided in Example \ref{ex:ex:nest2multi-orbit} has size  of order asymptotically equal to $q^{16k-2}$, as $k$ (or $q$) tends to infinity, larger than that constructed in Example \ref{ex:nest2one-orbit}. We remark that by the Johnson type bound II \eqref{eq:johnsonboundoptimal}, the maximum possible size for a code with the same parameters is asymptotically equal to $q^{16k}$, as $k$ (or $q$) tends to infinity.

\begin{example} \label{ex:ex:nest2multi-orbit}
Consider the codes $\cC_1$ and $\cC_2$ as in Example \ref{ex:Largeoneorbits}. 
As proved in \cite[Theorem 3.5]{zhang2024large} (without considering the orbits defined by different $\omega^i$'s),
\[
\cC_1=\bigcup_{a_1 \in \F_{q^k}} \Orb(V_{1,a_1})=\bigcup_{\alpha \in \F_{q^{3k}}^*} \bigcup_{a_1 \in \F_{q^k}} (\alpha V_{1,a_1}) \subseteq  \mathcal{G}_q(3k,k)
\]
is a multi-orbit cyclic subspace code having size $t_1=\lvert \C_1 \rvert= \frac{q^k(q^{3k}-1)}{q-1}$ and minimum distance $2k-2$ and 
\[
\cC_2=\bigcup_{a_2 \in \F_{q^{3k}}} \Orb(V_{2,a_2})=\bigcup_{\alpha \in \F_{q^{9k}}^*} \bigcup_{a_2 \in \F_{q^{3k}}} (\alpha V_{2,a_2}) \subseteq \mathcal{G}_q(9k,3k)
\]
is a multi-orbit cyclic subspace code having size $t_2=\lvert \C_2 \rvert= \frac{q^{3k}(q^{9k}-1)}{q-1}$ and minimum distance $2(3k)-2$.
Then by \Cref{cor:C1C2optimalimpliesC2odotC1optimal}, $\C_2 \odot \C_1$ is a multi-orbit cyclic subspace code in $\mathcal{G}_{q}(9k,k)$ having size $\lvert \C_1 \rvert \lvert \C_2 \rvert=\frac{q^{4k}(q^{3k}-1)(q^{9k}-1)}{(q-1)^2}$ and minimum distance $2k-2$.

$\hfill \lozenge$
\end{example}

\subsection{Operation $\odot$ between more than two multi-orbit cyclic subspace codes}
The operation $\odot$ can be extended to more than two multi-orbit cyclic subspace codes as follows. \\
%So, if we have a multi-orbit cyclic subspace code 
%$\C_3=\bigcup_{i=1}^{s_3} \Orb(V_i) \subseteq \mathcal{G}_q(r_1r_2r_3k,r_1r_2k)$
%we can define inductively 
%\[
%\C_3 \odot \C_2 \odot \C_1:=\C_3 \odot (\C_2 \odot \C_1)=\bigcup_{j_3=1}^{s_3} \bigcup_{j_2=1}^{t_2} \bigcup_{j_1=1}^{t_1} \Orb(V_{3,j_3} \odot U_{2,j_2} \odot U_{1,j_1})
%\]
%We can extend the operation $\odot$ for $e$ codes as follows.\\
Let $n=r_{e-1}\cdots r_1 m$, where  $r_1,\ldots,r_{e-1},m>1$ are positive integers. Consider
\begin{equation}
\cC_1=\bigcup_{h_1=1}^{s_1} \Orb(V_{1,h_1})=\{U_{1,1},\ldots,U_{1,t_1}\} \subseteq \mathcal{G}_q(m,k), 
\end{equation}
with $k>1$ and let %for any $i=2,\dots, e$ let
\begin{equation}
\label{eq:Ci}
\begin{aligned}
 \C_2&=\bigcup_{h_2=1}^{s_2} \Orb(V_{2,h_2})=\{U_{2,1},\ldots,U_{2,t_2}\} \subseteq \mathcal{G}_q(r_1m ,m)\\
 \C_3&=\bigcup_{h_3=1}^{s_3} \Orb(V_{3,h_3})=\{U_{3,1},\ldots,U_{3,t_3}\} \subseteq \mathcal{G}_q(r_2r_1m ,r_1m)\\
 \vdots\\
 \C_e&=\bigcup_{h_e=1}^{s_e} \Orb(V_{e,h_e})=\{U_{e,1},\ldots,U_{e,t_e}\} \subseteq \mathcal{G}_q(r_{e-1}\cdots r_1m, r_{e-2}\cdots r_1m)
\end{aligned}
\end{equation}
be full-length multi-orbit cyclic subspace codes, %Note that $\lvert \C_i\rvert=t_i$ and $t_i=s_i\frac{q^{r_1\dots r_ik}-1}{q-1}$ for every $i \in \{1,\ldots, e\}$. 
where, for every $i \in \{2,\ldots,e\}$ and for every $h_i \in \{1,\ldots,s_i\}$, \[V_{i,h_i}=\im(\Phi_{i,h_i}),\]
where 
\[
\Phi_{i,h_i}: \F_{q^{r_{i-2} \cdots r_1m}} \longrightarrow \F_{q^{r_{i-1} \cdots r_1 m}} 
\]
is an injective $\F_q$-linear map, with $r_0:=1$.

Define, recursively $\C_i \odot \C_{i-1} \odot \cdots \odot \C_1$ as follows. Recall 
\[
\cC_2 \odot \cC_1=\bigcup_{h_2=1}^{s_2}\bigcup_{j_1=1}^{t_1}\Orb(V_{2,h_2} \odot U_{1,j_1}) =\bigcup_{j_2=1}^{t_2}\bigcup_{j_1=1}^{t_1}(U_{2,j_2} \odot U_{1,j_1}) \subseteq \mathcal{G}_q(r_1m,k)
\]
and define for any $3\leq i \leq e$
\[
\C_{i} \odot \C_{i-1} \odot \cdots \odot \C_1:=\C_i \odot (\C_{i-1} \odot \cdots \odot \C_1)\subseteq \mathcal{G}_q(r_{i-1}\cdots r_1m,k).
\]
As in the case $e=2$, note that $\C=\C_{e} \odot \C_{e-1} \odot \cdots \odot \C_1$ turns out to be the code
\begin{equation}
\label{eq:multiorbitgeneralized}
\begin{array}{rl}
\cC & =\bigcup\limits_{h_e=1}^{s_e} \bigcup\limits_{j_{e-1}=1}^{t_{e-1}}\cdots \bigcup\limits_{j_1=1}^{t_1}\Orb(V_{e,h_e} \odot U_{e-1,j_{e-1}} \odot \cdots \odot U_{1,j_1}) \\
&=\bigcup\limits_{j_e=1}^{t_e} \bigcup\limits_{j_{e-1}=1}^{t_{e-1}}\cdots \bigcup\limits_{j_1=1}^{t_1}(U_{e,j_e} \odot U_{e-1,j_{e-1}} \odot \cdots \odot U_{1,j_1}).
\end{array}
\end{equation}

To better understand how $\odot$ acts on more than two cyclic subspace codes, we give the following example, which relies on the codes defined in   \Cref{ex:RRToneorbits}.

\begin{example} \label{ex:RRT3one-orbit}
Let $r_2=r_1=3,\,m=3k$ and $n=27k$, we will construct a multi-orbit code in $\mathcal{G}_q(27k,k)$ by using codes in $\mathcal{G}_q(3k,k)$, $\mathcal{G}_q(9k,3k)$ and $\mathcal{G}_q(27k,9k)$. Consider the subspaces
\[
V_{1}=\{v+v^q\gamma_1 \colon v \in \F_{q^k} \} \subseteq \mathcal{G}_q(3k,k)
\]
and 
\[
V_{2}=\{u+u^q\gamma_2 \colon u \in \F_{q^{3k}} \} \subseteq \mathcal{G}_q(9k,3k)
\]
as in \Cref{ex:RRToneorbits}, where $V_2=\im(\Phi_2)$ and \[
\begin{array}{llll}
\Phi_{2}: & \F_{q^{3 k}} & \longrightarrow & \F_{q^{9k}}\\
& u & \longmapsto & u+u^q\gamma_2
\end{array}
\] and let $\C_1$ and $\C_2$ be the corresponding codes. Now, consider $\gamma_3 \in \F_{q^{27k}}$ such that $\F_{q^{9k}}(\gamma_3)=\F_{q^{27k}}$. Define the subspace 
\[
V_{3}=\{w+w^q\gamma_3 \colon w \in \F_{q^{9k}} \} \subseteq \mathcal{G}_q(27k,9k)
\]
and define the code 
\[
\C_3=\Orb(V_{3}) \subseteq \mathcal{G}_q(27k,9k).
\]
The map 
\[
\begin{array}{llll}
\Phi_{3}: & \F_{q^{9 k}} & \longrightarrow & \F_{q^{27k}}\\
& w & \longmapsto & w+w^q\gamma_3
\end{array}
\]
is an injective $\F_q$-linear map
such that $V_{3}=\im(\Phi_{3})$.
Therefore, we can define 
\[
\begin{aligned}
  \cC_3 \odot \cC_2 \odot \cC_{1}&=\bigcup_{\beta \in \F_{q^{9k}}^*} \bigcup_{\alpha \in \F_{q^{3k}}^*} \Orb (V_{3} \odot (\beta V_2) \odot (\alpha V_{1})) \\
  &=\bigcup_{\xi\in\F_{q^{27k}}}\bigcup_{\beta \in \F_{q^{9k}}^*} \bigcup_{\alpha \in \F_{q^{3k}}^*}\{\xi(V_3\odot\lbrace \beta\alpha (v+v^q\gamma_1)+\beta\alpha^q (v+v^q\gamma_1)^q\gamma_2\colon v\in\F_{q^{k}}\rbrace)\}\\
  &=\bigcup_{\xi\in\F_{q^{27k}}}\bigcup_{\beta \in \F_{q^{9k}}^*} \bigcup_{\alpha \in \F_{q^{3k}}^*}\xi \{(\beta\alpha (v+v^q\gamma_1)+\beta\alpha^q (v+v^q\gamma_1)^q\gamma_2)+(\beta^q\alpha^q (v+v^q\gamma_1)^q\\
  &+\beta^q\alpha^{q^2} (v+v^q\gamma_1)^{q^2}\gamma_2^q)\gamma_3\colon v\in\F_{q^{k}}\}\subseteq\cG_q(27k,k).
\end{aligned}
\]

   $\hfill \lozenge$
\end{example}

Once defined the operation $\odot$ for $e$ multi-orbit cyclic subspace codes, with $e\geq 2$, we can generalize Corollary \ref{cor:C1C2optimalimpliesC2odotC1optimal}, allowing us to easily compute the size of the code $\cC_e\odot\dots\odot\cC_1$.

\begin{theorem} \label{th:nested>2}
%Assume that for every $j_2,j_2' \in \{1,\ldots,s_2\}$, with $j_2 \neq j_2'$ and for every $u,w,u',w' \in \F_{q^{r_1k}}$, the condition
%\[
%\Phi_{2,j_2}(u)\Phi_{2,j_2'}(u')=\Phi_{2,j_2}(w)\Phi_{2,j_2'}(w'),
%\]
%implies that $uv=u'v'$. 
Let $\C_i$ defined as in \eqref{eq:Ci} for $i=1,\dots,e$, $e\geq 2$, and assume that $d(\C_1)=2k-2$ and $d(\C_i)=2r_{i-2}\cdots r_1m-2$, for every $i\in \{2,\ldots,e\},$ with $r_0:=1$. The code $\C=\C_e \odot \cdots \odot \C_1 \subseteq \mathcal{G}_q( r_{e-1}\cdots r_1m,k)$ defined as in \eqref{eq:multiorbitgeneralized} is a multi-orbit cyclic subspace code having minimum distance $2k-2$ and size 
\[
\lvert \C \rvert = \lvert \C_e \odot \cdots \odot \C_1 \rvert =\prod_{i=1}^e\lvert \C_i \rvert.
%=\left(\prod_{i=1}^e s_i\frac{q^{r_{i-1}\cdots r_1 m}-1}{q-1}\right),
\]

\end{theorem}

\begin{proof}
We proceed by using finite induction on $i$, with $2\leq i\leq e$. When $i=2$, the assertion coincides with \Cref{cor:C1C2optimalimpliesC2odotC1optimal}. So assume that the assertion is true for every $j \leq i-1$ and we prove the assertion for $3\leq i \leq e$. By inductive hypothesis, $\C_{i-1}\odot \cdots \odot \C_1 \subseteq \mathcal{G}_q(r_{i-2}\cdots r_{1}m,k)$ has minimum distance $2k-2$. Moreover, the minimum distance of $\C_i \subseteq \mathcal{G}_q(r_{i-1}\cdots r_{1}m,r_{i-2}\cdots r_{1}m)$ is $2r_{i-2}\cdots r_{1}m-2$. Let $\overline{r}_1=r_{i-2}\cdots r_{1}$. Hence, $\C_{i-1}\odot \cdots \odot \C_1 \subseteq \mathcal{G}_q(\overline{r}_1m,k)$ has minimum distance $2k-2$ and $\C_i \subseteq \mathcal{G}_q(r_{i-1}\overline{r}_1m,\overline{r}_1m)$ has minimum distance $2\overline{r}_1m-2$. The assertion follows by applying \Cref{cor:C1C2optimalimpliesC2odotC1optimal}, to the codes $\C_{i-1}\odot \cdots \odot \C_1$ and $\C_i$. %, and so by \Cref{lm:equivalenceoptimalsubspaces}, we have that 
%\begin{equation} \label{eq:cond2codoptimal}
%\dim_{\F_q}\left(\left(U_{i-1,j_{i-1}^{\phantom{x}}} \odot \cdots \odot U_{1,j_{1}}\right) \cap \left(U_{i-1,j_{i-1}'} \odot \cdots \odot U_{1,j_{1}'}\right)\right)  \leq 1,
%\end{equation}
%for every $j_i,j_i' \in \{1,\ldots,t_i\}$, with $(j_{1},\ldots,j_{i-1})\neq (j_{1}',\ldots,j_{i-1}')$. 
%    Moreover, the code minimum distance of $\C_i \subseteq \mathcal{G}_q(r_1\ldots r_ik,r_1\ldots r_{i-1}k)$ is $2r_1\ldots r_ik-2$ and so by \Cref{lm:equivalenceoptimalsubspaces}, we have that 
%\begin{equation} \label{eq:cond1codoptimalmultiple}
%\dim_{\F_q}(U_{i,j_i} \cap U_{i,j_i'})  \leq 1,
%\end{equation}
%for every $j_i,j_i' \in \{1,\ldots,t_i\}$, with $j_i \neq j_i'$. Let $\overline{r}_1=r_1\ldots r_{i-1}$ and $\overline{r}_2=r_1\ldots r_{i-1}$

%We need to prove that when $(j_{1},\ldots,j_{i-1},j_i)\neq (j_{1}',\ldots,j_{i-1}',j_i')$, then 
%    \begin{equation} \label{eq:optimal2nested}
%     \dim_{\F_q}(U_{2,j_2} \odot U_{1,j_1} \cap  U_{2,j_2'} \odot U_{1,j_1'}) \leq 1. 
%    \end{equation}
\end{proof}

In the following example, we observe that the code constructed in \Cref{ex:RRT3one-orbit} has minimum distance $2k-2$ and size asymptotically equal to $q^{39k-3}$, as $k$ (or $q$) tends to infinity. 

\begin{example} \label{ex:nest3one-orbit}
Let $\cC_1,\cC_2,\cC_3$ as in Example \ref{ex:RRT3one-orbit}. As proved in \cite[Theorem 12]{roth2017construction}, $V_1$ is a Sidon space of $\cG_q(3k,k)$, $V_2$ is a Sidon space of $\cG_q(9k,3k)$ and $V_3$ is a Sidon space of $\cG_q(27k,9k)$. Then, by Theorem \ref{thm:sidoniffoptimal}, $\C_1$ is a cyclic subspace code having size $t_1=\lvert \C_1 \rvert= \frac{q^{3k}-1}{q-1}$ and minimum distance $2k-2$; $\C_2$ is a cyclic subspace code having size $t_2=\lvert \C_2 \rvert= \frac{q^{9k}-1}{q-1}$ and minimum distance $2(3k)-2$ and $\C_3$ is a cyclic subspace code having size $t_3=\lvert \C_3 \rvert= \frac{q^{27k}-1}{q-1}$ and minimum distance $2(9k)-2$.
So, by using \Cref{th:nested>2}, we get that 
$\cC_3 \odot \cC_2 \odot \cC_{1}$
is a multi-orbit cyclic subspace code in $\mathcal{G}_{q}(27k,k)$ having size $\lvert \C_1 \rvert \lvert \C_2 \rvert \lvert \C_3 \rvert =\frac{(q^{3k}-1)(q^{9k}-1)(q^{27k}-1)}{(q-1)^3}$ and minimum distance $2k-2$.

   $\hfill \lozenge$
\end{example}

In Example \ref{ex:nest3one-orbit} we obtain a code in $\cG_q(27k,k)$ whose size is asymptotically equal to $q^{36k-3}$, whereas the Johnson type bound II \eqref{eq:johnsonboundoptimal} is asymptotically equal to $q^{52k}$. In the following section, we will apply \Cref{th:nested>2} in order to get cyclic subspace codes with the same parameters and larger size.

\begin{comment}
\begin{example} \label{ex:nest3multi-orbit}
Consider the codes $\cC_1,\cC_2,\cC_3$ as in Example \ref{ex:Large3multi-orbit}. As proved in \cite[Theorem 3.5]{zhang2024large} (without considering the orbits defined by different $\omega^i$'s), $\C_1$ is a multi-orbit cyclic subspace code having size $t_1=\lvert \C_1 \rvert= \frac{q^k(q^{3k}-1)}{q-1}$ and minimum distance $2k-2$; $\C_2$ is a multi-orbit cyclic subspace code having size $t_2=\lvert \C_2 \rvert= \frac{q^{3k}(q^{9k}-1)}{q-1}$ and minimum distance $2(3k)-2$ and
$\C_3$ is a multi-orbit cyclic subspace code having size $t_3=\lvert \C_3 \rvert= \frac{q^{9k}(q^{27k}-1)}{q-1}$ and minimum distance $2(9k)-2$.
%Then by \Cref{th:nested>2} we get that $\cC_3\odot \cC_2\odot \cC_1$ 
\[
\C_3 \odot \C_2 \odot \C_{1}=\bigcup_{a_3 \in \F_{q^{9k}}} \bigcup_{\beta \in \F_{q^{9k}}^*} \bigcup_{a_2 \in \F_{q^{3k}}} \bigcup_{\alpha \in \F_{q^{3k}}^*} \bigcup_{a_1 \in \F_{q^k}} \Orb (V_{3,a_3} \odot (\beta V_{2,a_2}) \odot (\alpha V_{1,a_1})).
\]
is a multi-orbit cyclic subspace code in $\mathcal{G}_{q}(n,k)$ having size $\lvert \C_1 \rvert \lvert \C_2 \rvert \lvert \C_3 \rvert=\frac{q^{13k}(q^{3k}-1)(q^{9k}-1)(q^{27k}-1)}{(q-1)^3}$ minimum distance $2k-2$.
\end{example}
\end{comment}

\section{Special Constructions}
\label{sec:specialconstr}

In this section, we combine constructions of cyclic subspace codes %given in \cite{roth2017construction} and \cite{zhang2024large}, combining them 
in a suitable way using the operation $\odot$, in order to get new constructions whose sizes are larger than that of all the known constructions for almost of all the parameters that we are considering.
Due to the definition of the operation $\odot$, we need a tower of field extensions to apply the procedure of the previous section. Therefore, we divide the discussion according to the prime power decomposition of $n/k$. More precisely, we will distinguish the following cases:
\begin{itemize}
    \item $n=rk$, with $r$ power of $2$;
    \item $n=rk$, with $r$ odd;
    \item $n=rk$, with $r$ even and not a power of 2.
\end{itemize}

\subsection{Case $n=rk$, with $r$ power of $2$} \label{sec:fromRRT}
 The starting point is the construction of multi-orbit cyclic subspace codes provided in \cite{roth2017construction}.\\
In what follows, we will use the following notation:
\begin{itemize}
    \item  the pair $(r,k)$ in ${_r}V_{k}$ indicates that the considered subspace ${_r}V_{k}$ is in $\cG_q(rk,k)$ and the $3$-ple $(r,k,h)$ in the notation ${_r}V_{k,h}$ indicates that ${_r}V_{k,h}$ is the $h$-th element of the sequence of $({_r}V_{k})_h$ of subspaces of the type ${_r}V_{k}\subseteq \cG_q(rk,k)$;
     \item in the same spirit, the pair $(r,k)$ in ${_r}\Phi_{k}$ indicates that the considered $\fq$-linear map ${_r}\Phi_{k}$ maps elements of $\fqk$ into elements of $\F_{q^{rk}}$, i.e. ${_r}\Phi_{k}:\F_{q^{k}}  \longrightarrow  \F_{q^{rk}}$. The $3$-ple $(r,k,h)$ in the notation ${_r}\Phi_{k,h}$ denotes that the considered $\fq$-linear map ${_r}\Phi_{k,h}$ is the $h$-th element of the sequence $({_r}\Phi_{k})_h$ of maps of the type ${_r}\Phi_{k}:\F_{q^{k}}  \longrightarrow  \F_{q^{rk}}$;
    \item the pair $(r,k)$ in ${_r}\cC_{k}$ indicates that the considered code ${_r}\cC_{k}$ is a subset of $\cG_q(rk,k)$.
\end{itemize}

\begin{construction} [see \textnormal{\cite[Construction 37]{roth2017construction}}]
\label{constr:multiorbitRothRavivTamo}
For any prime power $q\geq 3$ and positive integer $k$, let $\omega$ be a primitive element of $\fqk$. For $n=2k$, let $\gamma \in\fqn \setminus \F_{q^k}$ such that $\N_{q^n/q^k}(\gamma)=\omega$. For $h=1,\dots, \tau$, where $\tau=\lfloor\frac{q-1}{2}\rfloor$, define
\begin{equation} \label{eq:subspaceroth}
{_2}V_{k,h}:=\lbrace v+v^q\omega^h\gamma\colon v\in\fqk\rbrace\subseteq \F_{q^{2k}}=\fqn
\end{equation}
and let \[
{_2}\C_{k}=\bigcup_{h=1}^{\tau}\mathrm{Orb}({_2}V_{k,h}) \subseteq \mathcal{G}_q(2k,k).\] 
\end{construction}

In \cite{roth2017construction}, it is shown that the code ${_2}\cC_{k}$ as in \Cref{constr:multiorbitRothRavivTamo} is a cyclic subspace code of minimum distance $2k-2$, as follows.

\begin{theorem} [see \textnormal{\cite[Lemma 38]{roth2017construction}}]\label{th:multiorbitRothRavivTamo}
Let $n=2k$, $q\geq 3$ be a prime power and $\tau=\lfloor\frac{q-1}{2}\rfloor$. The code ${_2}\cC_{k}\subseteq\cG_q(n,k)$ as in \Cref{constr:multiorbitRothRavivTamo} is a cyclic subspace code of cardinality $\tau\frac{q^{n}-1}{q-1}$ and minimum distance $2k-2$.
\end{theorem}

Let $q,\omega,\gamma,\tau$ be as in the hypothesis of \Cref{constr:multiorbitRothRavivTamo}.
Define for every $h \in \{1,\ldots,\tau\}$, the following injective $\F_q$-linear map
\begin{equation} \label{eq:mapRoth}
\begin{array}{llll}
{_2}\Phi_{k,h}: & \F_{q^{k}} & \longrightarrow & \F_{q^{2k}} \\
& v & \longmapsto & v+v^q\omega^h \gamma. 
\end{array}
\end{equation}
Thus
\[
{_2}V_{k,h}=\im({_2}\Phi_{k,h}).
\]

Note that the map ${_2}\Phi_{k,h}$ also depends on $\omega$ and $\gamma$. Therefore, if we change $k$ to a positive integer $k'$, we are implicitly assuming that we are simultaneously considering an element $\omega'$ to be a primitive element of $\F_{q^{k'}}$ and an element $\gamma'\in \F_{q^{2k'}}\setminus\F_{q^{k'}}$ such that $\N_{q^{2k'}/q^{k'}}(\gamma')=\omega'$. Therefore, we can employ the \Cref{constr:multiorbitRothRavivTamo}, with different values of $k$, along with \Cref{th:nested>2}, to obtain new multi-orbit cyclic subspace codes, as shown in what follows.

\begin{construction} 
\label{constr:nestedRRT}
Let $q\geq 3$ be a prime power, $n=2^ek$, with $e\geq 2$. For every $i=1,\dots,e$, let $\omega_i$ be primitive element of $\F_{q^{2^{i-1}k}}$ and let  $\gamma_i\in\F_{q^{2^ik}}\setminus\F_{q^{2^{i-1}k}}$ be such that $\mathrm{N}_{{q^{2^ik}}/q^{2^{i-1}k}}(\gamma_i)=\omega_i$.
So, consider
\[
{_2}V_{2^{i-1}k,h_i}=\lbrace v+v^q\omega_i^{h_i}\gamma_i \colon v\in\F_{q^{2^{i-1}k}} \rbrace \subseteq \F_{q^{2^{i}k}},
\]
for $h_i=1,\dots,\tau$, where $\tau=\lfloor\frac{q-1}{2}\rfloor$, and let
\[
{_2}\C_{2^{i-1}k}=\bigcup_{h_i=1}^{\tau}\mathrm{Orb}({_2}V_{2^{i-1}k,h_i})\subseteq \mathcal{G}_q(2^{i}k,2^{i-1}k).
\]
Define the code
\begin{equation} \label{eq:combiningeRoth}
{_{2^e}}\mathcal{C}_k:=\bigodot_{i=1}^{e}{_2}\cC_{2^{i-1}k}\subseteq\cG_q(2^ek,k).
\end{equation}
\end{construction}

We now show that the code ${_{2^e}}\cC_k\subseteq\cG_q(2^ek,k)$ in Construction \ref{constr:nestedRRT} is a cyclic subspace code of minimum distance $2k-2$.

\begin{theorem}
\label{thm:SpecialRRT}
Let $n=2^ek$, with $e\geq 2$ and let $q\geq 3$ be a prime power. The code ${_{2^e}}\mathcal{C}_k \subseteq \mathcal{G}_q(n,k)$ as in \Cref{constr:nestedRRT}
is a multi-orbit cyclic subspace code in $\mathcal{G}_q(2^ek,k)$ of cardinality 
\[
\lvert {_{2^e}}\mathcal{C}_k \rvert = \prod\limits_{i=1}^e \lvert {_2}\C_{2^{i-1}k} \rvert =\prod\limits_{i=1}^e \left( \left\lfloor\frac{q-1}{2}\right\rfloor\frac{q^{2^ik}-1}{q-1} \right)
\]
and minimum distance $2k-2$.     
\end{theorem}

\begin{proof}
We prove the assertion by using \Cref{th:nested>2}. First, we note that for any $i=1,\dots,e$ \[
{_2}V_{2^{i-1}k,h_i}=\im({_2}\Phi_{2^{i-1}k,h_i}),
\] where $\Phi_{2^{i-1}k,h_i}$ denotes the following injective $\F_q$-linear map
\begin{equation}
\begin{array}{llll}
{_2}\Phi_{2^{i-1}k,h_i}: & \F_{q^{2^{i-1}k}} & \longrightarrow & \F_{q^{2^ik}} \\
& v & \longmapsto & v+v^q\omega_i^{h_i} \gamma_i. 
\end{array}
\end{equation}
Therefore, the code ${_{2^e}}\mathcal{C}_k =\bigodot_{i=1}^{e}\left({_2}\C_{2^{i-1}k}\right)$ is well-defined. Moreover, by Theorem \ref{th:multiorbitRothRavivTamo}, we know that for every $i=1,\ldots,e$
\[
\lvert {_2}\cC_{2^{i-1}k}  \rvert = t_i=\tau \frac{q^{2^{i}k}-1}{q-1}
\]
and ${_2}\C_{2^{i-1}k}$ has minimum distance $2(2^{i-1}k)-2$. Thus by Theorem \ref{th:nested>2} we get that ${_{2^e}}\cC_k$ is a cyclic subspace code of cardinality $t_1\cdots t_e=\lfloor\frac{q-1}{2}\rfloor\prod\limits_{i=1}^e \frac{q^{2^ik}-1}{q-1}$ and minimum distance $2k-2$.
\end{proof}

\begin{remark} \label{rk:mapsC2e}
By construction, we have
\[
{_{2^e}}\mathcal{C}_k=\bigodot_{i=1}^{e}\left( \bigcup_{h_i=1}^{\tau}\mathrm{Orb}({_2}V_{2^{i-1}k,h_i}) \right).
\]
Thus, by \eqref{eq:multiorbitgeneralized}, any representative of the orbits of ${_{2^e}}\mathcal{C}_k$ is of the form
\[
{_2}V_{2^{e-1}k,h_e}\odot ((\alpha_{e-1}){_2}V_{2^{e-2}k,h_{e-1}}) \odot \dots \odot  ((\alpha_{1}){_2}V_{k,h_1}),
\]
for some $\alpha_i \in \mathbb{F}_{q^{2^{i}k}}$, with $i= {1, \ldots, e-1}$.
On the other hand, for each $i=1,\dots, e$, we have ${_2}V_{2^{i-1}k,h_i}=\im({_2}\Phi_{2^{i-1}k,h_i})\subseteq \mathbb{F}_{q^{2^ik}}$ and ${_2}\Phi_{2^{i-1}k,h_i}$ is injective. Therefore, any representative of the orbits of ${_{2^e}}\mathcal{C}_k$ is the image of the injective $\F_q$-linear map 
\begin{equation}
\label{eq:mapperoth2e}
{_{2^e}}\Phi_{k}:={_2}\Phi_{2^{e-1}k,h_e}\circ ((\alpha_{e-1}){_2}\Phi_{2^{e-2}k,h_{e-1}}) \circ \dots \circ  ((\alpha_{1}){_2}\Phi_{k,h_1}):\F_{q^k} \longrightarrow \F_{q^{2^ek}}.  
\end{equation}
\end{remark}

The following example explicitly shows the subspaces of the codes ${_{2^e}}\cC_k$ provided in \Cref{th:multiorbitRothRavivTamo} when $e=2$.

\begin{example}
\label{ex:asymptoptimaln=4k}
Let $n=4k$ and $\tau=\lfloor\frac{q-1}{2}\rfloor$. For $i=1,2$, let $\omega_i$ be a primitive element of $\F_{q^{2^{i-1}k}}$. Let $\gamma_i\in\F_{q^{2^ik}}\setminus\F_{q^{2^{i-1}k}}$ such that $\mathrm{N}_{{q^{2^ik}}/q^{2^{i-1}k}}(\gamma_i)=\omega_i$. Then define
\[
{_2}V_{k,h_1}:=\lbrace u+u^q\omega_1^{h_1}\gamma_1\colon u\in\fqk\rbrace\subseteq \F_{q^{2k}},
\]
\[
{_2}V_{2k,h_2}:=\lbrace v+v^q\omega_2^{h_2}\gamma_2\colon v\in\F_{q^{2k}}\rbrace\subseteq \F_{q^{4k}}.
\]
Consider
\[
_{2}\cC_k=\bigcup_{h_1=1}^{\tau}\mathrm{Orb}(_{2}V_{k,h})\subseteq\cG_q(2k,k)
\]
and 
\[
_{2}\cC_{2k}=\bigcup_{h_2=1}^{\tau}\mathrm{Orb}(_{2}V_{2k,h})\subseteq\cG_q(4k,2k).
\]
By Theorem \ref{thm:SpecialRRT}, the code
\[
\begin{aligned}
_{4}\cC_{k} =_{2}\cC_{2k}\odot_{2}\cC_{k}&=\bigcup_{h_2=1}^{\tau} \bigcup_{\alpha \in \F_{q^{2k}}^*}
\bigcup_{h_1=1}^{\tau}\mathrm{Orb}(_{2}V_{2k,h_2}\odot \alpha (_{2}V_{k,h_1}))\\
 &=\bigcup_{h_2=1}^{\tau} \bigcup_{\alpha \in \F_{q^{2k}}^*}
\bigcup_{h_1=1}^{\tau}\Orb(\lbrace v+v^q\omega_2^{h_2}\gamma_2\colon v\in\alpha (_{2}V_{k,h_1}))\rbrace)\\
 &=\bigcup_{h_2=1}^{\tau} \bigcup_{\alpha \in \F_{q^{2k}}^*}
\bigcup_{h_1=1}^{\tau}\Orb(\lbrace (\alpha u+\alpha u^q\omega_1^{h_1}\gamma_1)+(\alpha u+\alpha u^q\omega_1^{h_1}\gamma_1)^q\omega_2^{h_2}\gamma_2\colon u\in\fqk\rbrace)
 \end{aligned}
\]
is a cyclic subspace code in $\mathcal{G}_q(4k,k)$ of cardinality $(\lfloor\frac{q-1}{2}\rfloor)^2\cdot\frac{q^{2k}-1}{q-1}\frac{q^{4k}-1}{q-1}$ and minimum distance $2k-2$.

$\hfill \lozenge$
\end{example}

Note that the size of the code in Example \ref{ex:asymptoptimaln=4k} is asymptotically equal to $\frac{1}{4}q^{6k}$ and so it asymptotically approaches the Johnson type bound II within a factor of $\frac{1}{4}+o_k(1)$ ( or $\frac{1}{4}+o_q(1))$). In Section \ref{sec:comparison}, we will see that this behavior of the size of the code ${_{2^e}\cC_k}\subseteq\cG_q(2^ek,k)$ holds true for any $e\geq 2$.

\subsection{Case $n=rk$, with $r$ odd}

In Section \ref{sec:fromRRT}, we have considered the case $n=rk$, where $r$ is a power of $2$. In this section, we construct multi-orbit cyclic subspace codes with $n=rk$, for $r$ an odd integer, by employing again  \Cref{th:nested>2}. Here, the starting point will be the construction of multi-orbit cyclic subspace codes provided in \cite{zhang2024large}.\\
We use the following notation. For a positive integer $\ell_0$, and an element $\underline{\alpha}=(\alpha_1,\ldots,\alpha_{\ell_0}) \in \F_{q^k}^{\ell_0}$, define 
\[
P_{k,\underline{\alpha}}(\gamma)=1+\alpha_1\gamma+\cdots+\alpha_{\ell_0}\gamma^{\ell_0}.
\]
and we set $P_{k,\underline{\alpha}}(\gamma):=1$ for $\ell_0=0$. We set
\[
S_{k,\ell_0}:=\{(\alpha_1,\ldots,\alpha_{\ell_0}) \in \F_{q^k}^{\ell_0}\, \colon \,\alpha_1,\ldots,\alpha_{\ell_0-1}\in \F_{q^k},\alpha_{\ell_0} \in \F_{q^k}^*\}.
\]

Given a positive integer $\ell$ with $\ell_0\leq \ell$, we define the following set of indices 
\[
I_{k,\ell_0,\ell}:=S_{k,\ell_0} \times \F_{q^k} \times \{0,1,\dots, q-2\}\times \{\ell_0,\ell_0+1,\dots,\ell\}.
\]
Then
    \begin{itemize}
    \item ${_r}V_{k,\underline{h}}$ denotes a subspace in $\cG_q(rk,k)$ which depends on a vector of indices $\underline{h}=(\underline{\alpha_{\ell_0}},a,j,\ell_1)\in I_{k,\ell_0,\ell}$, i.e.
    \[\underline{\alpha_{\ell_0}}\in S_{k,\ell_0},\,  a\in\fqk,\, 0\leq j\leq q-2, \,\ell_0\leq \ell_1 \leq \ell \,\, \text{for}\, \,\ell:=\left\lceil \frac{r}{2} \right\rceil-2.
    \]
    \item ${_r}\Phi_{k,\underline{h}}:  \F_{q^{k}}  \longrightarrow  \F_{q^{rk}}$ denotes an $\fq$-linear map which depends on a vector of indices $\underline{h}=(\underline{\alpha_{\ell_0}},a,j,\ell_1)\in I_{k,\ell_0,\ell}$, where $\ell:=\left\lceil \frac{r}{2} \right\rceil-2.$
%\[\underline{\alpha_{\ell_0}}\in S_{k,\ell_0},\,  a\in\fqk,\, 0\leq j\leq q-2, \,\ell_0\leq \ell_1 \leq \ell \,\, \text{for}\, \,\ell:=\left\lceil \frac{r}{2} \right\rceil-2.
%    \]
    \item ${_r}\cC_{k,\ell_0}$ denotes a multi-orbit code in $\cG_q(rk,k)$ which depends on a positive integer $\ell_0\leq \ell=\lceil \frac{r}{2}\rceil-1$.
    \item ${_r}\cC_{k}$ denotes a multi-orbit code in $\cG_q(rk,k)$ which is union over $\ell_0$ of codes denoted by ${_r}\cC_{k,\ell_0}$.
\end{itemize}

We consider the construction of cyclic subspace codes defined in \textnormal{\cite[Theorem 3.5]{zhang2024large}}. This construction is defined in $\cG_q(rk,k)$ for any integer $r\geq 3$. However, for our goal, we will only consider the case in which $r$ is an odd prime.

\begin{construction} [see \textnormal{\cite[Theorem 3.5]{zhang2024large}}]
\label{constr:constrLarge}
Let $n=pk$, for some positive integers $k \geq 2$ and an odd prime $p$. Let $\omega$ be a primitive element of $\fqk$. Let 
$\ell:=\left\lceil \frac{n}{2k} \right\rceil-2=\frac{p-3}{2}$ and let $\gamma\in\fqn$ be such that $\fqk(\gamma)=\fqn$. Let $0\leq \ell_0\leq \ell$ and define 
\begin{equation}
    \label{eq:subspacesplarge}
{_p}V_{k,\underline{h}}
=\{uP_{k,\underline{\alpha_{\ell_0}}}(\gamma)+\omega^j(u^q+au)\gamma^{\ell_1+1}\colon u \in \F_{q^k}\}\in\cG_q(n,k),
\end{equation}
for any $\underline{h}=(\underline{\alpha_{\ell_0}},a,j,\ell_1) \in I_{k,\ell_0,\ell}$. %i.e. $\underline{\alpha_{\ell_0}} \in S_{k,{\ell_0}}$, $a \in \F_{q^k}$, $j \in [q-2]$ and $\ell_0\leq \ell_1\leq \ell$.\\
Let
\[
{_p}\cC_{k,\ell_0}:=\bigcup_{\underline{h}\in I_{k,\ell_0,\ell}}\mathrm{Orb}(_{p}V_{k,\underline{h}})%=\bigcup_{\underline{\alpha_{\ell_0}}\in S_{k,{\ell_0}}}
%\bigcup_{h=0}^{q-2} \bigcup_{a\in\fqk} \bigcup_{\ell_1=\ell_0}^{\ell} \Orb({_r}V_{k,(\underline{\alpha_{\ell_0}},i,a,\ell_1)})\subseteq\cG_q(n,k),
\]
and define the code
\begin{equation} \label{eq:codesrfromlarge}
{_p}\cC_k:=\bigcup_{\ell_0=0}^{\ell}{_p}\cC_{k,\ell_0}\subseteq\cG_q(n,k).
\end{equation}
\end{construction}

The authors of \cite{zhang2024large} proved that the code ${_{p}}\cC_k\subseteq\cG_q(pk,k)$ in Construction \ref{constr:constrLarge} is a cyclic subspace code of minimum distance $2k-2$.

\begin{theorem} [see \textnormal{\cite[Theorem 3.5]{zhang2024large}}]
\label{th:Large}
Let $n=pk$, for some positive integers $k\geq 2$ and $p$ odd prime, and let $\ell=\frac{p-3}{2}$. The code ${_p}\cC_k$ as in \Cref{constr:constrLarge} is a cyclic subspace code in $\cG_q(n,k)$ with size $\frac{q^k(q^{(\ell+1)k}-1)(q^{n}-1)}{q^k-1}$ and minimum distance $2k-2$.
\end{theorem}

The subspaces $_pV_{k,\underline{h}}$ of the Construction \ref{constr:constrLarge} can be seen as images of injective $\F_q$-linear maps. This will allow us to use the operation $\odot$ to construct new multi-orbit cyclic subspace codes as described in what follows. Let $q,\omega,\gamma,\ell$ as in the hypothesis of \Cref{constr:constrLarge}. Note that $\{1,\gamma,\dots,\gamma^{p-1}\}$ is an $\fqk$-basis of $\fqn$. 
Define for every $\underline{h}=(\underline{\alpha_{\ell_0}},a,j,\ell_1)  \in I_{k,\ell_0,\ell}$, the following injective $\F_q$-linear map
\begin{equation} \label{eq:mapLarge}
\begin{array}{llll}
{_p}\Phi_{k,\underline{h}}: & \F_{q^{k}} & \longrightarrow & \F_{q^{pk}} \\
& u & \longmapsto & uP_{k,\underline{\alpha_{\ell_0}}}(\gamma)+\omega^j(u^q+au)\gamma^{\ell_1+1}
\end{array}
\end{equation}
and note that
\[
{_p}V_{k,\underline{h}}=\im({_p}\Phi_{k,\underline{h}}).
\]

As for the maps ${_2}\Phi_{k,h}$ defined in \eqref{eq:mapRoth}, the maps ${_p}\Phi_{k,\underline{h}}$ also depend on $\omega$ and $\gamma$. Thus, even in this case, if we change $k$ to a positive integer $k'$, we implicitly assume that we are simultaneously considering an element $\omega'$ to be a primitive element of $\F_{q^{k'}}$, an element $\gamma'\in \F_{q^{rk'}}$ such that $\F_{q^{rk'}}=\F_{q^{k'}}(\gamma')$, where $r$ is an odd number. 

Now, we can use \Cref{constr:constrLarge}, with different values of $k$, along with \Cref{th:nested>2}, to obtain new multi-orbit cyclic subspace codes. 
\\
Assume that $n=rk$, with $r$ a composite odd number. We will consider two cases, the first one is $r=p^e$ where $p$ is an odd prime number. The second is the case $r=p_1^{e_1}\cdots p_z^{e_z}$, where $p_1<\ldots<p_z$ are distinct primes and $e_1,\ldots,e_z \geq 1$. The idea is now to consider several multi-orbit codes in $\cG_q(n_i,k_i)$ defined as in \Cref{constr:constrLarge}, where $n_i/k_i=p_i$ is odd, and then we will combine them using \Cref{th:nested>2}. \\
\\
Let start with the case $n=p^ek$, with $p$ and odd prime. 

\begin{construction} \label{constr:usinglargenestedpequal}
Let $n=p^ek$, with $p$ a prime, $p\neq 2$, $k\geq 2$. Let $\ell:=\frac{p-3}{2}$ and for any $i=1,\dots, e$, let $\omega_{i}$ be a primitive element of $\F_{q^{p^{i-1}k}}$ and $\gamma_i\in\F_{q^{p^ik}}$ such that $\F_{q^{p^{i-1}k}}(\gamma_i)=\F_{q^{p^ik}}$. 
For any $i=1,\dots, e$, define 
\[{_p}V_{p^{i-1}k,\underline{h_i}}
=\{uP_{k,\underline{\alpha_{\ell_{0,i}}}}(\gamma_i)+\omega_{i}^{j_i}(v^q+a_iv)\gamma_i^{\ell_{1,i}+1}\colon v \in \F_{q^{p^{i-1}k}}\}\in\cG_q(p^{i}k,p^{i-1}k),
\]
for any $\underline{h}_i=(\underline{\alpha_{\ell_0,i}},a_i,j_i,\ell_{1}) \in I_{p^{i-1}k,\ell_{0},\ell}$,  as in \eqref{eq:subspacesplarge}.
%i.e. $\underline{\alpha_{\ell_0,i}} \in S_{k,{\ell_{0,i}}}$, $a \in \F_{q^{p^{i-1}k}}$, $j_i \in \{0,1,\dots,q-2\}$ and $\ell_{0,i}\leq \ell_{1,i}\leq \ell_i$ for any $i=1,\dots,e$.\\ 
Let
\[
{_p}\cC_{p^{i-1}k,\ell_{0}}=\bigcup_{\underline{h}_i\in I_{p^{i-1}k,\ell_{0},\ell}}\mathrm{Orb}(_{p}V_{p^{i-1}k,\underline{h_i}})
%\bigcup_{\underline{\alpha_{\ell_{0,i}}}\in S_{p^{i-1}k,{\ell_{0,i}}}}
%\bigcup_{h=0}^{q-2} \bigcup_{a\in\fqk} \bigcup_{\ell_{1,i}=\ell_{0,i}}^{\ell_i} \Orb({_p}V_{p^{i-1}k,\underline{\alpha_{\ell_{0,i}}},h_i,a_i,\ell_{1,i}})\subseteq\cG_q(p^{i}k,p^{i-1}k),
\]
and consider
\[
{_p}\cC_{p^{i-1}k}=\bigcup_{\ell_{0}=0}^{\ell_i}{_p}\cC_{p^{i-1}k,\ell_{0}}
\]
for any $i=1,\dots,e$.\\
Define
\begin{equation} \label{eq:codesnestedlargepe}
{_{p^{e}}}\cC_k:=\bigodot_{i=1}^{e}{_p}\cC_{p^{i-1}k}\subseteq \cG_q(n,k).
\end{equation} 
\end{construction}

We prove that the code ${_{p^e}}\cC_k\subseteq\cG_q(p^ek,k)$ in Construction \ref{constr:usinglargenestedpequal} is a cyclic subspace code of minimum distance $2k-2$.

\begin{theorem}
\label{thm:theoremlargewithn=pek}
Let $n=p^ek$ for some positive integer $k \geq 2$,  with $p$ an odd prime, and let $\ell:=\frac{p-3}{2}$. The code ${_{p^{e}}}\cC_k$ as in   \Cref{constr:usinglargenestedpequal}, is a multi-orbit cyclic subspace code in $\cG_q(n,k)$ of cardinality \[|{_{p^{e}}}\cC_k|=\prod\limits_{i=1}^e \left( \frac{q^{p^{i-1}k}(q^{(\ell+1)p^{i-1}k}-1)(q^{p^ik}-1)}{q^{p^{i-1}k}-1} \right)
\] and minimum distance $2k-2$.
\end{theorem}

\begin{proof}
We prove the assertion by using \Cref{th:nested>2}. We note that for every $i=1,\dots,e$ \[
{_p}V_{p^{i-1}k,\underline{h_i}}=\im({_p}\Phi_{p^{i-1}k,\underline{h_i}}),
\] where ${_p}\Phi_{p^{i-1}k,\underline{h_i}}$ denotes the following injective $\F_q$-linear map
\begin{equation}
\begin{array}{llll}
{_p}\Phi_{p^{i-1}k,\underline{h_i}}: & \F_{q^{p^{i-1}k}} & \longrightarrow & \F_{q^{p^ik}} \\
& u & \longmapsto & uP_{k,\underline{\alpha_{\ell_{0},i}}}(\gamma_i)+\omega_i^{h_i}(u^q+a_iu)\gamma_i^{\ell_{1}+1}, 
\end{array}
\end{equation}
defined as in \eqref{eq:mapLarge}.
Therefore, the code ${_{p^e}}\mathcal{C}_k:=\bigodot_{i=1}^{e}{_{p}}\cC_{p^{i-1}k}$ is well-defined. \\ Moreover, by Theorem \ref{th:Large}, we know that for every $i=1,\ldots,e$
\[
\left\lvert {_{p}}\cC_{p^{i-1}k}  \right\rvert = t_i=\frac{q^{p^{i-1}k}(q^{(\ell+1)p^{i-1}k}-1)(q^{p^ik}-1)}{q^{p^{i-1}k}-1}
\]
and ${_p}\C_{p^{i-1}k}$ has minimum distance $2(p^{i-1}k)-2$. Thus by Theorem \ref{th:nested>2} we get that ${_{p^e}}\cC_k$ is a cyclic subspace code of cardinality $\prod_{i=1}^e\left(\frac{q^{p^{i-1}k}(q^{(\ell+1)p^{i-1}k}-1)(q^{p^ik}-1)}{q^{p^{i-1}k}-1}\right)$ and minimum distance $2k-2$.
\end{proof}

\begin{remark} \label{rk:mapsCpe}
In the same spirit of \Cref{rk:mapsC2e}, we note that any representative of the orbits of ${_{p^e}}\mathcal{C}_k$ is the image of the injective $\F_q$-linear map 
\begin{equation}
\label{eq:mappepe}
{_p}\Phi_{p^{e-1}k,\underline{h_e}} \circ ((\alpha_{e-1}){_p}\Phi_{p^{e-2}k,\underline{h_{e-1}}}) \circ \dots \circ  ((\alpha_{1}){_p}\Phi_{k,\underline{h_1}}):\F_{q^k} \longrightarrow \F_{q^{p^ek}},  
\end{equation}
for some $\alpha_i \in \mathbb{F}_{q^{p^{i}k}}$, with $i \in {1, \ldots, e-1}$.
\end{remark}

In the following example, we show Construction \ref{constr:usinglargenestedpequal} in the case $n=27k$.

\begin{example}
\label{ex:n=27k}
Let $n=27k$, with $k\geq 2$. Then $p=3$ and $\ell=0$. For $i \in \{1,2,3\}$, let  $\omega_{i}$ be a primitive element of $\F_{q^{3^{i-1}k}}$ and $\gamma_i\in\F_{q^{3^ik}}$ such that $\F_{q^{3^{i-1}k}}(\gamma_i)=\F_{q^{3^ik}}$ . Define 
\[{_3}V_{k,\underline{h_1}}
=\{uP_{k,\underline{\alpha_{\ell_{0,1}}}}(\gamma_1)+\omega_{1}^{j_1}(v^q+a_1v)\gamma_1^{\ell_{1}+1}\colon v \in \F_{q^{k}}\}\in\cG_q(3k,k),
\]
\[
{_3}V_{3k,\underline{h_2}}
=\{uP_{k,\underline{\alpha_{\ell_{0,2}}}}(\gamma_2)+\omega_{2}^{j_2}(v^q+a_2v)\gamma_i^{\ell_{1}+1}\colon v \in \F_{q^{3k}}\}\in\cG_q(9k,3k),
\]
\[
{_3}V_{9k,\underline{h_3}}
=\{uP_{k,\underline{\alpha_{\ell_{0,3}}}}(\gamma_3)+\omega_{3}^{j_3}(v^q+a_3v)\gamma_3^{\ell_{1}+1}\colon v \in \F_{q^{9k}}\}\in\cG_q(27k,9k),
\]
where $\underline{h}_i=(\underline{\alpha_{\ell_0,i}},a_i,j_i,\ell_{1,i}) \in I_{p^{i-1}k,\ell_{0},\ell}$ for any $i=1,2,3$ as in \eqref{eq:subspacesplarge}.
Let
\[
{_3}\cC_{k,\ell_{0}}=\bigcup_{\underline{h}_1\in I_{k,\ell_{0},\ell}}\mathrm{Orb}(_{3}V_{k,\underline{h_1}})\subseteq\cG_q(3k,k)
\]
\[
{_3}\cC_{3k,\ell_{0}}=\bigcup_{\underline{h}_2\in I_{3k,\ell_{0},\ell}}\mathrm{Orb}(_{3}V_{3k,\underline{h_2}})\subseteq\cG_q(9k,3k)
\]
\[
{_3}\cC_{9k,\ell_{0}}=\bigcup_{\underline{h}_3\in I_{9k,\ell_{0},\ell}}\mathrm{Orb}(_{3}V_{9k,\underline{h_3}})\subseteq\cG_q(27k,9k).
\]
Then consider
\[
{_3}\cC_{k}=\bigcup_{\ell_{0}=0}^{\ell}{_3}\cC_{k,\ell_{0}}
\]
\[
{_3}\cC_{9k}=\bigcup_{\ell_{0}=0}^{\ell}{_3}\cC_{3k,\ell_{0}}
\]
\[
{_3}\cC_{9k}=\bigcup_{\ell_{0}=0}^{\ell}{_3}\cC_{9k,\ell_{0}}.
\]
Then by Theorem \ref{thm:theoremlargewithn=pek}, the code
\begin{equation} 
{_{27}}\cC_k:={_3}\cC_{9k}\odot{_3}\cC_{3k}\odot {_3}\cC_{k} \subseteq \cG_q(27k,k)
\end{equation}
has minimum distance $2k-2$ and cardinality $
q^{13k}(q^{3k}-1)(q^{9k}-1)(q^{27k}-1).$
   $\hfill \lozenge$
\end{example}

Note that the size of the code in Example \ref{ex:n=27k} is asymptotically equal to $q^{52k}$ and so it asymptotically reaches the Johnson type bound II. In Section \ref{sec:comparison}, we will see that this behavior of the size of the code ${_{3^e}\cC_k}\subseteq\cG_q(3^ek,k)$ holds true for any $e\geq 2$.

In what follows, we will show that \Cref{constr:usinglargenestedpequal} can be extended to the case in which $n=p_1^{e_1}\dots p_z^{e_z}k$ with $p_1,\dots,p_z$ different odd primes. By using \Cref{th:Large} and \Cref{th:nested>2}, we will obtain, also in this case, large cyclic subspace codes with minimum distance $2k-2$.

As observed in \Cref{rk:mapsCpe}, for every odd prime $p$, the representatives of the orbits of the codes $_{p^e}\mathcal{C}_k$ are defined by injective $\mathbb{F}_q$-linear maps. This will allow us to apply the operation $\odot$ by combining codes of type ${_{p^e}}\mathcal{C}_k$, with different values of the prime $p$. Also, we note that in this case we can consider different towers of field extensions and this choice will affect the size of the resulting code. We start with an example.

\begin{example}
\label{ex:usinglargenesteddifferentprimes}
%Then $\ell_{1}=1,\,\ell_2=0$. Let $\omega_{1,1}$ be a primitive element of $\F_{q^{k}}$ and let $\gamma_{1,1}\in\F_{q^{5k}}$ be such that $\F_{q^{k}}(\gamma_{1,1})=\F_{q^{5k}}$. Let $\omega_{2,1}$ be a primitive element of $\F_{q^{5k}}$.
%Let $\gamma_{2,1}\in\F_{q^{15k}}$ be such that $\F_{q^{5k}}(\gamma_{2,1})=\F_{q^{15k}}$.  Let $\omega_{2,2}$ be a primitive element of $\F_{q^{15k}}$.
%Let $\gamma_{2,2}\in\F_{q^{45k}}$ be such that $\F_{q^{15k}}(\gamma_{2,2})=\F_{q^{45k}}$.
Let $n=45k$. We can consider the following different towers of field extensions.

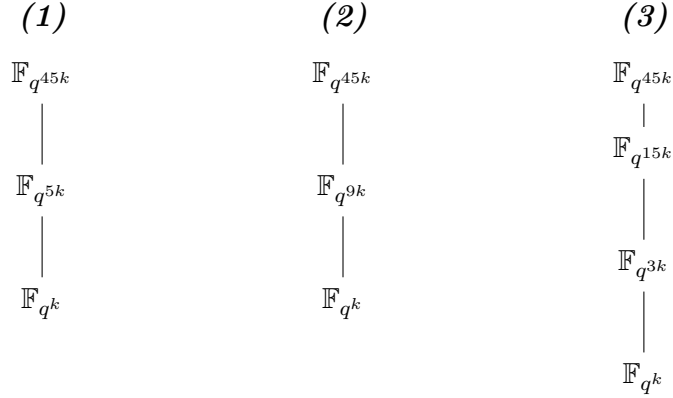
\begin{center}
\begin{tikzpicture}[node distance=1.6cm]

% Etichette colonne in alto
\node (Col1) at (0,5.3) {\large\textbf{(1)}};
\node (Col2) at (4,5.3) {\large\textbf{(2)}};
\node (Col3) at (8,5.3) {\large\textbf{(3)}};

% Livello F_{q^{45k}}
\node (A3) at (0,4.5) {$\mathbb{F}_{q^{45k}}$};
\node (B3) at (4,4.5) {$\mathbb{F}_{q^{45k}}$};
\node (C4) at (8,4.5) {$\mathbb{F}_{q^{45k}}$};

% Primo reticolo (colonna 1)
\node (A2) at (0,3) {$\mathbb{F}_{q^{5k}}$};
\node (A1) at (0,1.5) {$\mathbb{F}_{q^k}$};
\draw (A1) -- (A2) -- (A3);

% Secondo reticolo (colonna 2)
\node (B2) at (4,3) {$\mathbb{F}_{q^{9k}}$};
\node (B1) at (4,1.5) {$\mathbb{F}_{q^k}$};
\draw (B1) -- (B2) -- (B3);

% Terzo reticolo (colonna 3)
\node (C3) at (8,3.5) {$\mathbb{F}_{q^{15k}}$};
\node (C2) at (8,2)   {$\mathbb{F}_{q^{3k}}$};
\node (C1) at (8,0.5) {$\mathbb{F}_{q^k}$};
\draw (C1) -- (C2) -- (C3) -- (C4);

\end{tikzpicture}
\end{center}

\begin{comment}
\begin{center}
\begin{tikzpicture}[scale=3,  
    every node/.style={font=\footnotesize},
    redarrow/.style={-, red, thick},
    bluearrow/.style={-, blue, thick},
    greenarrow/.style={-, green!60!black, thick},
    blackarrow/.style={->, black, thick}
]

% Vertices
\coordinate (A) at (0,0);         % Π_q
\coordinate (B) at (1,0.5);       % Π̄_q^{3k}
\coordinate (C) at (1,1.5);       % Π̄_q^{9k}
\coordinate (D) at (0,2);         % Π̄_q^{45k}
\coordinate (E) at (-1,1.5);      % Π̄_q^{15k}
\coordinate (F) at (-1,0.5);      % Π̄_q^{5k}

% Nodes
\node at ($(A)+(0,-0.1)$) {\(\F_{q^k}\)};
\node at ($(B)+(0.2,-0.05)$) {\(\F_{q^{5k}}\)};
\node at ($(C)+(0.2,0)$) {\(\F_{q^{15k}}\)};
\node at ($(D)+(0,0.15)$) {\(\F_{q^{45k}}\)};
\node at ($(E)+(-0.2,0)$) {\(\F_{q^{9k}}\)};
\node at ($(F)+(-0.2,0)$) {\(\F_{q^{3k}}\)};

% Arrows (colored)
\draw[redarrow] (D) -- (E);
\draw[redarrow] (E) -- (A);

\draw[greenarrow] (D) -- (B);
\draw[greenarrow] (B) -- (A);

\draw[bluearrow] (D) -- (C);
\draw[bluearrow] (C) -- (F);
\draw[bluearrow] (F) -- (A);

\end{tikzpicture}
\end{center}   
\end{comment}

Consider the tower $\textbf{(1)}$. To define a code in $\cG_q(45k,k)$ by using Theorem \ref{th:nested>2}, we can apply the operation $\odot$  between the code ${_5}\mathcal{C}_{k} \subseteq \mathcal{G}_q(5k,k)$ defined as in \eqref{eq:codesrfromlarge} and the code ${_9}\mathcal{C}_{5k} \subseteq \mathcal{G}_q(45k,5k)$ defined as in \eqref{eq:codesnestedlargepe}. Indeed, as noted in \Cref{rk:mapsCpe}, any representative of the orbits of ${_{9}}\mathcal{C}_{5k}$ is the image of the injective $\mathbb{F}_{q^5}$-linear map from $\mathbb{F}_{q^{5k}} \rightarrow \mathbb{F}_{q^{45k}}$ as in \eqref{eq:mappepe}. Therefore by Theorem \ref{th:nested>2}, we can define the code
\[
\cC:={_{9}}\cC_{5k}\odot{_{5}}\cC_{k}\subseteq \cG_q(45k,k)
\] 
and, using Theorems~\ref{thm:theoremlargewithn=pek} and~\ref{th:Large}, this code has minimum distance $2k-2$ and cardinality \[
\begin{array}{rl}
\lvert {_{9}}\cC_{5k}\odot{_{5}}\cC_{k} \rvert & = \lvert {_{9}}\cC_{5k}\rvert \lvert {_{5}}\cC_{k} \rvert\\
& = \left( q^{20k}(q^{15k}-1)(q^{45k}-1) \right) \left(\frac{q^k(q^{2k}-1)(q^{5k}-1)}{q^k-1}\right)\sim q^{87k}
\end{array}
\] 
as $k\to+\infty$ (or as $q\to+\infty$).
\\
Similarly, consider the tower $\textbf{(2)}$
and the corresponding codes ${_9}\mathcal{C}_{k} \subseteq \mathcal{G}_q(9k,k)$ and ${_5}\mathcal{C}_{9k} \subseteq \mathcal{G}_q(45k,9k)$ defined as in \Cref{constr:usinglargenestedpequal}.  By Remark \ref{rk:mapsCpe} and by Theorem \ref{th:nested>2}, we can define the code 
\[
\cC:={_5}\cC_{9k} \odot {_9}\cC_{k} \subseteq \cG_q(45k,k)
\]
and, using again Theorems~\ref{thm:theoremlargewithn=pek} and~\ref{th:Large}, it has minimum distance $2k - 2$ and cardinality \[ \begin{array}{rl}
\lvert {_5}\cC_{9k} \odot {_9}\cC_{k}  \rvert & = \lvert {_5}\cC_{9k}\rvert \lvert {_9}\cC_{k} \rvert\\
& =\left(\frac{q^{9k}(q^{18k}-1)(q^{45k}-1)}{q^{9k}-1}\right) \left( q^{4k}(q^{3k}-1)(q^{9k}-1) \right) \\
& = q^{13k} (q^{18k}-1)(q^{45k}-1)(q^{3k}-1)\sim q^{79k}
\end{array} 
\]
as $k\to+\infty$ (or as $q\to+\infty$).\\
Finally, we can consider the tower $\textbf{(3)}$ and the corresponding codes ${_{3}}\cC_{k}\subseteq \cG_q(3k,k),\,{_{5}}\cC_{3k}\subseteq \cG_q(15k,3k), \,{_{3}}\cC_{15k}\subseteq\cG_q(45k,15k)$ defined as in \Cref{constr:constrLarge}. By Remark \ref{rk:mapsCpe} and \Cref{th:nested>2}, we can define the code
\[
\cC={_{3}}\cC_{15k} \odot {_{5}}\cC_{3k} \odot  {_{3}}\cC_{k} \subseteq \cG_q(45k,k)
\]
which has, by applying Theorem~\ref{th:Large}, minimum distance $2k - 2$ and cardinality \[
\begin{array}{rl}
\lvert {_{3}}\cC_{15k} \odot {_{5}}\cC_{3k} \odot  {_{3}}\cC_{k} \rvert & = \lvert {_{3}}\cC_{15k}\rvert \lvert {_{5}}\cC_{3k} \rvert \lvert {_{3}}\cC_{k} \rvert  \\
& =(q^{15k}(q^{45k}-1)) \left(\frac{q^{3k}(q^{6k}-1)(q^{15k}-1)}{q^{3k}-1}\right) \left( q^{k}(q^{3k}-1) \right) \\
& = q^{19k}(q^{45k}-1) (q^{6k}-1)(q^{15k}-1)\sim q^{85k}
\end{array}
\]
as $k\to+\infty$ (or as $q\to+\infty$).\\

We emphasize that among the three cases (towers \textbf{(1)}, \textbf{(2)}, and \textbf{(3)}), the construction from tower~\textbf{(1)} yields the code with the largest size.
$\hfill \lozenge$

\end{example}

%Note that, for any $i=1,\dots, e$, $\im({_p}\Phi_{p^{i-1}k,\underline{h_i}})\subseteq \mathbb{F}_{q^{p^ik}}$ and ${_p}\Phi_{p^{i-1}k,\underline{h_i}}$ is injective. Thus we can define the map ${_{p^e}}\Phi_k\colon \mathbb{F}_{q^{k}}\longrightarrow \mathbb{F}_{q^{p^ek}}$ as
%\begin{equation}
%\label{eq:mappelargediffprimes}
%{_{p^e}}\Phi_k:={_p}\Phi_{p^{e-1}k,\underline{h_e}}\circ\dots\circ{_p}\Phi_{pk,\underline{h_2}}\circ{_p}\Phi_{k,\underline{h_1}},    
%\end{equation}
%and it is injective as well.

We are ready to state our construction for $n=p_z^{e_z}\cdots p_1^{e_1}k$.

\begin{construction}
\label{constr:usinglargenesteddifferentprimes}
Let $n=p_z^{e_z}\cdots p_1^{e_1}k$, with $p_1>\ldots>p_z$ distinct odd primes. For any $j=1,\dots,z$ and for every $i_j \in \{1,\ldots,e_j\}$, 
let $\ell_{j}:=\frac{p_{j}-3}{2}$,  $\omega_{j,i_j}$ be a primitive element of $\F_{q^{p_j^{i_j-1}p_{j-1}^{e_{j-1}}\dots p_1^{e_1}k}}$ and let $\gamma_{j,i_j}\in\F_{q^{p_j^{i_j}p_{j-1}^{e_{j-1}}\dots p_1^{e_1}k}}$ be such that $\F_{q^{p_j^{i_j}p_{j-1}^{e_{j-1}}\dots p_1^{e_1}k}}=\F_{q^{p_j^{i_j-1}p_{j-1}^{e_{j-1}}\dots p_1^{e_1}k}}(\gamma_{j,i_j})$.
Then consider the following codes defined as in \eqref{eq:codesnestedlargepe}
\begin{equation} 
\label{eq:codesnestedlargedifferentprimes}
\begin{aligned}
{_{p_1^{e_1}}}\cC_{k}&:=\bigodot_{i_1=1}^{e_1}{_{p_1}}\cC_{p_1^{i_1-1}k}\subseteq \cG_q(p_1^{e_1}k,k)\\
{_{p_2^{e_2}}}\cC_{p_1^{e_1}k}&:=\bigodot_{i_2=1}^{e_2}{_{p_2}}\cC_{p_2^{i_2-1}p_1^{e_1}k}\subseteq \cG_q(p_2^{e_2}p_1^{e_1}k,p_1^{e_1}k)\\
\vdots\\
{_{p_z^{e_z}}}\cC_{p_{z-1}^{e^{z-1}}\dots p_1^{e_1}k}&:=\bigodot_{i_z=1}^{e_z}{_{p_z}}\cC_{p_z^{i_z-1}\dots p_1^{e_1}k}\subseteq \cG_q(p_z^{e_z}p_{z-1}^{e_{z-1}}\dots p_1^{e_1}k,p_{z-1}^{e_{z-1}}\dots p_1^{z_1}k).
\end{aligned}
\end{equation} Then define
\begin{equation} \label{eq:codesnestedlargedifferentcomplete}
{_{p_z^{e_z}\dots p_1^{e_1}}}\cC_k:={_{p_z^{e_z}}}\cC_{p_{z-1}^{e^{z-1}}\dots p_1^{e_1}k}\odot \dots\odot  {_{p_{2}^{e_{2}}}}\cC_{p_{1}^{e^{1}}k}\odot{_{p_1^{e_1}}}\cC_{k}\subseteq \cG_q(n,k).
\end{equation}
\end{construction}

We prove that the code ${_{p_z^{e_z}\dots p_1^{e_1}}}\cC_k\subseteq\cG_q(p_z^{e_z}\dots p_1^{e_1}k,k)$ as in \Cref{constr:usinglargenesteddifferentprimes} has minimum distance $2k-2$ and large size.\\ 
In the following, we adopt the convention that \(\prod_{i=a}^{b} f(i) = 1\) when \(b < a\).

\begin{theorem}
\label{thm:largediffprimesodd}
Let $n=p_z^{e_z}\cdots p_1^{e_1}k$, with $p_1>\ldots>p_z$ distinct odd primes. The code ${_{p_z^{e_z}\dots p_1^{e_1}}}\cC_k$ as in  \Cref{constr:usinglargenesteddifferentprimes},  is a multi-orbit cyclic subspace code in $\cG_q(n,k)$ of cardinality
\begin{equation}
\label{eq:cardinalitydifferentprimeslarge}
\left\lvert{_{p_z^{e_z}\dots p_1^{e_1}}}\cC_k\right\rvert=\prod\limits_{j=1}^{z}\prod\limits_{i_j=1}^{e_j} \left( \frac{q^{p_j^{{i_j}-1}\prod\limits_{h=1}^{j-1}p_{h}^{e_{h}}k}(q^{(\ell_j+1)p_j^{{i_j}-1}\prod\limits_{h=1}^{j-1}p_{h}^{e_{h}}k}-1)(q^{p_j^{i_j}\prod\limits_{h=1}^{j-1}p_{h}^{e_{h}}k}-1)}{q^{p_j^{i_j-1}\prod\limits_{h=1}^{j-1}p_{h}^{e_{h}}k}-1} \right)
\end{equation}
and minimum distance $2k-2$.
\end{theorem}
\begin{proof}
By Theorem \ref{thm:theoremlargewithn=pek}, we know that \[
\left\lvert{_{p_1^{e_1}}}\cC_{k}\right\rvert =\prod\limits_{i_1=1}^{e_1} \left( \frac{q^{p_1^{i_1-1}k}(q^{(\ell_1+1)p^{i_1-1}k}-1)(q^{p_1^{i_1}k}-1)}{q^{p^{i_1-1}k}-1} \right)
\]
and for every $j=2,\dots,z$ 
\[
\left\lvert{_{p_j^{e_j}}}\cC_{p_{j-1}^{e_{j-1}}\dots p_1^{e_1}k}\right\rvert =\prod\limits_{i_j=1}^{e_j} \left( \frac{q^{p_j^{{i_j}-1}\prod\limits_{h=1}^{j-1}p_{h}^{e_{h}}k}(q^{(\ell_j+1)p_j^{{i_j}-1}\prod\limits_{h=1}^{j-1}p_{h}^{e_{h}}k}-1)(q^{p_j^{i_j}\prod\limits_{h=1}^{j-1}p_{h}^{e_{h}}k}-1)}{q^{p_j^{i_j-1}\prod\limits_{h=1}^{j-1}p_{h}^{e_{h}}k}-1} \right).
\]
Also, ${_{p_1^{e_1}}}\cC_{k}$ has minimum distance $2k-2$ and ${_{p_j^{e_j}}}\cC_{p_{j-1}^{e^{j-1}}\dots p_1^{e_1}k}$ has minimum distance $2(p_{j-1}^{e^{j-1}}\dots p_1^{e_1}k)-2$ for $j=2,\dots,z$. Thus by Theorem \ref{th:nested>2} we get that ${_{p_z^{e_z}\dots p_1^{e_1}}}\cC_k$ is a cyclic subspace code of minimum distance $2k-2$ and cardinality $\left\lvert{_{p_z^{e_z}\dots p_1^{e_1}}}\cC_k\right\rvert=\left\lvert{_{p_z^{e_z}}}\cC_{p_{z-1}^{e^{z-1}}\dots p_1^{e_1}k}\right\rvert\dots \left\lvert{_{p_2^{e_2}}}\cC_{ p_1^{e_1}k}\right\rvert\left\lvert{_{p_1^{e_1}}}\cC_{k}\right\rvert,$
thus  \eqref{eq:cardinalitydifferentprimeslarge} follows.
\end{proof}

\begin{remark}
    It is worth pointing out that, using the same strategy as in \Cref{constr:usinglargenesteddifferentprimes}, and as noted in \Cref{ex:usinglargenesteddifferentprimes}, we can also combine the codes $_{p_i^{e_i}}\mathcal{C}_k$ by employing a different ordering of the primes $p_i$. This leads to new codes, which may have different sizes or could potentially be inequivalent. A detailed study of this remains an interesting direction for future work. However, as we will prove in \Cref{sec:comparison}, with the particular ordering we have chosen, the sizes of our codes ${_r}\cC_k\subseteq\cG_1(rk,k)$ exhibit faster asymptotic growth than those of all previously known cyclic subspace codes with minimum distance $2k-2$.
\end{remark}

%In what follows, we explicity show, as example, Construction \ref{constr:usinglargenesteddifferentprimes} for $n=45k$.

\subsection{Case $n=rk$, with $r$ even and not a power of 2}
\label{sec:fromlarge}
Finally, we consider the case $n=rk$, where $r$ is even and $r$ is not a power of 2.\\
Let $n=2^ep_z^{e_z}\cdots p_1^{e_1}k$ where $p_1>\dots >p_z$ are distinct odd primes. In this case, we can construct codes in $\cG_q(n,k)$ by combining ${_{p_z^{e_z}\cdots p_1^{e_1}}}\cC_k \subseteq \mathcal{G}_q(p_z^{e_z}\dots p_1^{e_1}k,k)$ as in \eqref{eq:codesnestedlargedifferentcomplete}  with $_{2^e}\mathcal{C}_{p_z^{e_z}\dots p_1^{e_1}k}\subseteq \mathcal{G}_q(2^ep_z^{e_z}\dots p_1^{e_1}k, p_z^{e_z}\dots p_1^{e_1}k)$, as in \eqref{eq:combiningeRoth} when $e \geq 2$ and as in \eqref{th:multiorbitRothRavivTamo}, when $e=1$. Indeed, we know that the representatives of the orbits of the codes as in \eqref{eq:combiningeRoth} are defined by injective $\mathbb{F}_q$-linear maps, see \Cref{rk:mapsC2e}. Thus, we can apply the $\odot$ operation between codes of $_{2^e}\mathcal{C}_{p_z^{e_z}\dots p_1^{e_1}k}$ and ${_{p_z^{e_z}\dots p_1^{e_1}}}\cC_k$, as follows.

\begin{construction}
\label{constr:roth+large}
Let $n=2^ep_z^{e_z}\cdots p_1^{e_1}k$ where $p_1>\dots >p_z$ are distinct odd primes. Consider ${_{p_z^{e_z}\dots p_1^{e_1}}}\cC_k$ defined as in \eqref{eq:codesnestedlargedifferentcomplete} and $_{2^e}\cC_{p_z^{e_z}\dots p_1^{e_1}k}$ defined as in \eqref{eq:combiningeRoth}. Then define
\[
_{2^ep_z^{e_z}\dots p_1^{e_1}}\cC_{k}:=_{2^e}\cC_{p_z^{e_z}\dots p_1^{e_1}k}\odot{_{p_z^{e_z}\dots p_1^{e_1}}}\cC_k\subseteq\mathcal{G}_q(n,k).
\]
\end{construction}

We prove that the code $_{2^ep_z^{e_z}\dots p_1^{e_1}}\cC_{k}\subseteq\cG_q(2^ep_z^{e_z}\dots p_1^{e_1}k,k)$ as in \Cref{constr:roth+large} has minimum distance $2k-2$.

\begin{theorem}
\label{thm:roth+large}
Let $n=2^ep_z^{e_z}\cdots p_1^{e_1}k$, where $p_1>\ldots>p_z$ distinct odd primes. In the hypothesis of \Cref{constr:roth+large}, the code ${_{2^ep_z^{e_z}\dots p_1^{e_1}}}\cC_k$ is a multi-orbit cyclic subspace code in $\cG_q(n,k)$ of cardinality
\begin{equation}
\label{eq:cardinalitydifferentprimeslargeeroth}
\prod\limits_{j=1}^{z}\prod\limits_{i_j=1}^{e_j} \left( \frac{q^{p_j^{{i_j}-1}\prod\limits_{h=1}^{j-1}p_{h}^{e_{h}}k}(q^{(\ell_j+1)p_j^{{i_j}-1}\prod\limits_{h=1}^{j-1}p_{h}^{e_{h}}k}-1)(q^{p_j^{i_j}\prod\limits_{h=1}^{j-1}p_{h}^{e_{h}}k}-1)}{q^{p_j^{i_j-1}\prod\limits_{h=1}^{j-1}p_{h}^{e_{h}}k}-1} \right)\prod\limits_{i=1}^e \left( \left\lfloor\frac{q-1}{2}\right\rfloor\frac{q^{2^ip_{z}^{e^{z}}\dots p_1^{e_1}k}-1}{q-1} \right)
\end{equation}
and minimum distance $2k-2$.
\end{theorem}
\begin{proof}
By Theorem \ref{thm:largediffprimesodd}, we know that
\[
\left\lvert{_{p_z^{e_z}\dots p_1^{e_1}}}\cC_k\right\rvert=\prod\limits_{j=1}^{z}\prod\limits_{i_j=1}^{e_j} \left( \frac{q^{p_j^{{i_j}-1}\prod\limits_{h=1}^{j-1}p_{h}^{e_{h}}k}(q^{(\ell_j+1)p_j^{{i_j}-1}\prod\limits_{h=1}^{j-1}p_{h}^{e_{h}}k}-1)(q^{p_j^{i_j}\prod\limits_{h=1}^{j-1}p_{h}^{e_{h}}k}-1)}{q^{p_j^{i_j-1}\prod\limits_{h=1}^{j-1}p_{h}^{e_{h}}k}-1} \right)
\]
and by Theorem \ref{thm:SpecialRRT}
\[
\left\lvert{_{2^e}}\cC_{p_{z}^{e^{z}}\dots p_1^{e_1}k}\right\rvert =\prod\limits_{i=1}^e \left( \left\lfloor\frac{q-1}{2}\right\rfloor\frac{q^{2^ip_{z}^{e^{z}}\dots p_1^{e_1}k}-1}{q-1} \right).
\]
Also, ${_{p_z^{e_z}\dots p_1^{e_1}}}\cC_{k}$ has minimum distance $2k-2$ and ${_{2^e}}\cC_{p_{z}^{e^{z}}\dots p_1^{e_1}k}$ has minimum distance $2(p_{z}^{e^{z}}\dots p_1^{e_1})k-2$. Thus by Theorem \ref{th:nested>2} we get that ${_{2^ep_z^{e_z}\dots p_1^{e_1}}}\cC_k$ is a cyclic subspace code of minimum distance $2k-2$ and cardinality $\left\lvert{_{2^ep_z^{e_z}\dots p_1^{e_1}}}\cC_k\right\rvert\left\lvert{_{p_z^{e_z}\dots p_1^{e_1}}}\cC_k\right\rvert$, 
thus \eqref{eq:cardinalitydifferentprimeslargeeroth} follows.
\end{proof}

The following example shows Construction \ref{constr:roth+large} in the case $n=36k$.

\begin{example}
Let $n=36k=3^22^2k$. Consider ${_{9}}\cC_k$ defined as in \eqref{eq:codesnestedlargedifferentcomplete} and $_{4}\cC_{9k}$ defined as in \eqref{eq:combiningeRoth}. By Theorem \ref{thm:roth+large}, the code
\[
_{36}\cC_k:={_{4}}\cC_{9k}\odot{_{9}}\cC_{k}\subseteq\cG_{q}(36k,k)
\]
has cardinality $2k-2$ and size \[
%\frac{q^{3k}(q^{3k}-1)(q^{9k}-1)}{q^{3k}-1}\frac{q^{k}(q^{k}-1)(q^{3k}-1)}{q^{k}-1}  \left\lfloor\frac{q-1}{2}\right\rfloor^2\frac{q^{36k}-1}{q-1}\frac{q^{18k}-1}{q-1}=
q^{4k}(q^{3k}-1)(q^{9k}-1)\left\lfloor\frac{q-1}{2}\right\rfloor^2\frac{q^{18k}-1}{q-1}\frac{q^{36k}-1}{q-1}.
\]   $\hfill \lozenge$
\end{example}

\section{Comparison with known constructions of constant-dimension subspace codes and the Johnson type bound II}
\label{sec:comparison}

In this section, we compare the size of the codes constructed in \Cref{sec:specialconstr} with those of the constructions of subspace codes resumed in \Cref{table:constructionscyclic}, as well as with the Johnson type bound II for subspace codes. We summarize in Table \ref{table: Our} the constructions presented in Section \ref{sec:specialconstr}. For $n = rk$, depending on the prime factorization of $r$, we use the notation $_r\C_k$ to refer to a code in $\mathcal{G}_q(rk,k)$ obtained from Constructions \ref{constr:nestedRRT}, \ref{constr:usinglargenestedpequal}, \ref{constr:usinglargenesteddifferentprimes} and \ref{constr:roth+large}.

\begin{table}[htp]
	\centering
	\medskip
\small
\begin{tabular}{|c|c|c|c|} 
	\hline
Constructions & Parameters & Size & Reference\\ 
$_r\C_k$ & & & \\
	\hline
    & & & \\
$_{2^e}\C_k$ & $n=2^ek,$ & $\prod_{i=1}^e\left( \left\lfloor\frac{q-1}{2}\right\rfloor\frac{q^{2^ik}-1}{q-1} \right)$  & Theorem \ref{thm:SpecialRRT} \\
 & $q>2$& & \\
\hline 
&  & & \\
 $_{p^e}\C_k$ & $n=p^ek$ & $\prod\limits_{i=1}^e \left( \frac{q^{p^{i-1}k}(q^{\frac{(p-1)}2p^{i-1}k}-1)(q^{p^ik}-1)}{q^{p^{i-1}k}-1} \right)$ & Theorem \ref{thm:theoremlargewithn=pek} \\
& $p$ odd prime & & \\
\hline
& & & \\
$_{p_z^{e_z}\dots p_1^{e_1}}\C_k$& $n=p_z^{e_z}\cdots p_1^{e_1}k$ & $\prod\limits_{j=1}^{z}\prod\limits_{i_j=1}^{e_j} \left( \frac{q^{p_j^{{i_j}-1}\prod\limits_{h=1}^{j-1}p_{h}^{e_{h}}k}(q^{(\ell_j+1)p_j^{{i_j}-1}\prod\limits_{h=1}^{j-1}p_{h}^{e_{h}}k}-1)(q^{p_j^{i_j}\prod\limits_{h=1}^{j-1}p_{h}^{e_{h}}k}-1)}{q^{p_j^{i_j-1}\prod\limits_{h=1}^{j-1}p_{h}^{e_{h}}k}-1} \right)$ & Theorem \ref{thm:largediffprimesodd} \\
& $p_1>\cdots>p_z$  &  & \\
& odd primes & & \\
\hline 
& & & \\
$_{2^ep_z^{e_z}\dots p_1^{e_1}}\C_k$ & $n=2^ep_z^{e_z}\cdots p_1^{e_1}k$ & $\prod\limits_{j=1}^{z}\prod\limits_{i_j=1}^{e_j} \left( \frac{q^{p_j^{{i_j}-1}\prod\limits_{h=1}^{j-1}p_{h}^{e_{h}}k}(q^{(\ell_j+1)p_j^{{i_j}-1}\prod\limits_{h=1}^{j-1}p_{h}^{e_{h}}k}-1)(q^{p_j^{i_j}\prod\limits_{h=1}^{j-1}p_{h}^{e_{h}}k}-1)}{q^{p_j^{i_j-1}\prod\limits_{h=1}^{j-1}p_{h}^{e_{h}}k}-1} \right)$ & \Cref{thm:roth+large} \\
& $p_1>\cdots>p_z$ & & \\
&  odd primes & & \\
& $q>2$  & $\cdot\prod\limits_{i=1}^e \left( \left\lfloor\frac{q-1}{2}\right\rfloor\frac{q^{2^ip_{z}^{e^{z}}\dots p_1^{e_1}k}-1}{q-1} \right)$ & \\
\hline 
\end{tabular}
\caption{Constructions of  cyclic subspace codes in $\cG_q(n,k)$ with minimum distance $2k-2$ presented in this work.}\label{table: Our}
\end{table}

We start our analysis by describing the asymptotic growth of the size of the codes $_r\C_k \subseteq \mathcal{G}_q(rk,k)$, as $k$ or $q$ tends to infinity. 
\begin{itemize}
    \item \textbf{Case $n=2^ek$:} we refer to the codes $_{2^e}\C_k\subseteq \cG_q(2^ek,k)$, defined as in Construction \ref{constr:nestedRRT}. As a consequence of Theorem \ref{thm:SpecialRRT}, we have
    \begin{equation} \label{eq:asymptoticRTT+RTT}
   \lvert {_{2^e}}\mathcal{C}_k \rvert \sim \frac{1}{2^e} q^{2k+4k+\cdots+2^ek}=\frac{1}{2^e}q^{2k(2^e-1)}=\frac{1}{2^e}q^{2(n-k)}, \ \ \ \mbox{as } k \rightarrow + \infty \ \ \ (\mbox{or as } q \rightarrow + \infty).
   \end{equation}
   \item \textbf{Case $n=p^ek$, with $p$ odd prime:} we refer to the codes $_{p^e}\C_k\subseteq\cG_q(p^ek,k)$, defined as in Construction \ref{constr:usinglargenestedpequal}. By using Theorem \ref{thm:theoremlargewithn=pek}, we have
   \begin{equation} \label{eq:asymptoticsamelarge} \lvert _{p^e}\C_k \rvert \sim q^{\frac{1}{2}\sum\limits_{i=1}^e(3p^i-p^{i-1})k}=q^{\frac{(3p-1)(p^e-1)}{2(p-1)}k}, \ \ \ \mbox{ as }k \rightarrow + \infty \ \ \ (\mbox{or as }q \rightarrow + \infty)
    \end{equation}
    \item \textbf{Case $n=p_z^{e_z}\cdots p_1^{e_1}k$, with $p_1>\cdots>p_z$ odd primes:} in this case we refer to the codes $_{p_z^{e_z}\dots p_1^{e_1}}\C_k\subseteq\cG_q(p_z^{e_z}\cdots p_1^{e_1}k,k)$, defined as in Construction \ref{constr:usinglargenesteddifferentprimes}. As a consequence of \eqref{eq:asymptoticsamelarge} and Theorem \ref{thm:largediffprimesodd}, we get
    \begin{equation} \label{eq:asymptoticdifferentlarge}
    \begin{array}{rl}
      \lvert _{p_z^{e_z}\dots p_1^{e_1}}\C_k \rvert    &  =\left\lvert{_{p_z^{e_z}}}\cC_{p_{z-1}^{e^{z-1}}\dots p_1^{e_1}k}\right\rvert\dots \left\lvert{_{p_2^{e_2}}}\cC_{ p_1^{e_1}k}\right\rvert\left\lvert{_{p_1^{e_1}}}\cC_{k}\right\rvert \\
         & \sim q^{\sum\limits_{i=1}^z\left(\frac{(3p_i-1)(p_i^{e_i}-1)}{2(p_i-1)}\prod\limits_{j=1}^{i-1}p_j^{e_j}k\right)}, \ \ \ \mbox{ as }k \rightarrow + \infty \ \ \ (\mbox{or as }q \rightarrow + \infty)
    \end{array}
\end{equation}
    \item \textbf{Case $n=2^ep_z^{e_z}\cdots p_1^{e_1}k$, with $p_1>\cdots>p_z$ odd primes:} we are referring to the codes $_{2^ep_z^{e_z}\dots p_1^{e_1}}\C_k\subseteq\cG_q(2^ep_z^{e_z}\cdots p_1^{e_1}k,k)$, defined as in Construction \ref{constr:roth+large}. In this case, by using both \eqref{eq:asymptoticRTT+RTT} and \eqref{eq:asymptoticdifferentlarge}, together with Theorem \ref{thm:roth+large}, we obtain
    \begin{equation} \label{eq:asymptoticRTT+large}
    \begin{array}{rl}
    \lvert _{2^ep_z^{e_z}\dots p_1^{e_1}}\C_k \rvert & = \left\lvert{_{2^e}}\cC_{p_{z}^{e^{z}}\dots p_1^{e_1}k}\right\rvert \left\lvert{_{p_z^{e_z}\dots p_1^{e_1}}}\cC_{k}\right\rvert
    \\
    & \sim \frac{1}{2^e}q^{2\prod\limits_{i=1}^e p_i^{e_i}k(2^e-1)}q^{\sum\limits_{i=1}^z\left(\frac{(3p_i-1)(p_i^{e_i}-1)}{2(p_i-1)}\prod\limits_{j=1}^{i-1}p_j^{e_j}k\right)}, \ \ \ \mbox{ as }k \rightarrow + \infty \ \ \ (\mbox{or as }q \rightarrow + \infty)
    \end{array}
\end{equation}
\end{itemize}

\subsection{Comparison with known constructions}

Our study focuses on large multi-orbit cyclic subspace codes in $\mathcal{G}_q(n,k)$, where $n=rk$ and $r$ is a composite integer. In the next, we assume that $r \geq 6$ and $k\geq 2$. \\

We start first by comparing the constructions provided in \Cref{table:constructionscyclic}. By looking at Table \ref{table:constructionscyclic}, as $q$ or $k$ increases, $S_3(n,k,q), S_4(n,k,q)$ and $S_5(n,k,q)$ have the same asymptotic behavior.

%For our aim, we need to compare the Constructions 3), 4) and 5). 
%For Constructions 4) and 5) we assume $r=2h+1$, \textcolor{blue}{$h\geq 3?$} and note that in this case \[\ell+1=\left\lceil \frac{n}{2k}\right\rceil-1=h.\]  

\begin{comment}
\begin{lemma} \label{lm:comparing}
    Let $r=2 h+1$, with $h \geq 2$. Then
    \[S_3(n,k,q)= \Theta_k(S_4(n,k,q)), \Theta_k(S_5(n,k,q))\] and \[ S_3(n,k,q)= \Theta_q(S_4(n,k,q)), \Theta_q(S_5(n,k,q)).\] 
\end{lemma}

\begin{proof}
We have 
\begin{equation} \label{eq:infinitylarge}
    S_3(n,k,q) \sim q^{(\ell+1)k+n}, \ \ \ \mbox{as }k \rightarrow + \infty \ \ \ (\mbox{or as }q \rightarrow + \infty), 
    \end{equation}
\[S_4(n,k,q) \sim (\ell+1) q^{(\ell+1) k+n} , \ \ \ \mbox{as }k \rightarrow + \infty\ \ \ (\mbox{or as }q \rightarrow + \infty), \]
and 
\[S_5(n,k,q) \sim (\ell+1) q^{(\ell+1) k+n} , \ \ \ \mbox{as }k \rightarrow + \infty\ \ \ (\mbox{or as }q \rightarrow + \infty). \]
The assertion immediately follows.
\end{proof}
\end{comment}

Now, we compare the size of the new constructions provided in Section \ref{sec:specialconstr} with the sizes of the Constructions in Table \ref{table:constructionscyclic}. We will consider $n=rk$, where $r=p_z^{e_z}\cdots p_1^{e_1}$ for some primes $p_1>\cdots>p_z$. We assume $r\geq 6$ and $k \geq 2$. In the case where $r$ is even, our constructions also require $q\geq 3$.

In the following theorem, we prove that for any composite integer $r\geq 6$, and $n=rk$, the size of the codes $_r\C_k$ exhibits faster asymptotic growth compared to the codes in Construction 3) of Table \ref{table:constructionscyclic} and, therefore, when $r$ is odd, also compared to the codes in Constructions 4) and 5) of the same table. Consequently, when $r \geq 6$ the size of the codes $_r\C_k$ exhibits faster asymptotic growth than that of all known cyclic subspace codes in $\cG_q(n,k)$ with minimum distance $2k-2$.

\begin{theorem} \label{th:confrontoconlealtre}
Let $r\geq 6$ be a composite integer. For any prime power $q$ or, if $r$ is even, for any prime power $q\geq 3$, 
  \[
  S_3(rk,k,q)=o_k(\lvert _r\C_k \rvert) \ \ \ \mbox{ and }\ \ \ S_3(rk,k,q)=o_q(\lvert _r\C_k \rvert).
  \]
  If $r$ is odd, then
  \[
    S_4(rk,k,q),S_5(rk,k,q)=o_k(\lvert _r\cC_k \rvert) \ \ \ \mbox{ and }\ \ \ S_3(rk,k,q), S_5(rk,k,q)=o_q(\lvert _r\cC_k \rvert).
  \] In particular, for $k$ (or $q$) sufficiently large, the size of the code $_r\mathcal{C}_k \subseteq \mathcal{G}_q(rk,k)$, defined as in Table \ref{table: Our}, is larger than the sizes of all the Constructions in Table \ref{table:constructionscyclic}. 
\end{theorem}
\begin{proof}
We divide the analysis according to the following cases.\\
\underline{\textbf{Case 1: $n/k$ power of $2$.}} In this case $n=2^ek$, with $e \geq 3$. We have that:
\begin{equation} \label{eq:asympotyclarge2^e}
S_3(2^ek,k,q)
%=q^k\frac{q^{(2^{e-1}-1)k}-1}{q^k-1}(q^{2^ek}-1)
\sim q^{(2^{e}+2^{e-1}-1)k}, \ \ \  \mbox{ as }k \rightarrow + \infty \ \ \ (\mbox{or as }q \rightarrow + \infty)
\end{equation}

Therefore, comparing \eqref{eq:asymptoticRTT+RTT} with \eqref{eq:asympotyclarge2^e}, we get
\[
S_{3}(2^ek,k,q)=o_k(\lvert {_{2^e}}\mathcal{C}_k \rvert) \ \ \ \mbox{ and } \ \ \ S_{3}(2^ek,k,q)=o_q(\lvert {_{2^e}}\mathcal{C}_k \rvert).
\]

\underline{\textbf{Case 2: $n/k$ power of an odd prime $p$.}} In this case $n=p^ek$, where $p$ is an odd prime and $e \geq 2$. We observe that $\ell+1=\lceil\frac{p^e}{2}\rceil -1$, hence:
\begin{equation} \label{eq:asympotyclargep^e}
S_3(p^ek,k,q)
%=q^k\frac{q^{\frac{p^e-1}{2}k}-1}{q^k-1}(q^{p^ek}-1)
\sim q^{\frac{3p^e-1}{2}k}, \ \ \  \mbox{ as }k \rightarrow + \infty \ \ \ (\mbox{or as }q \rightarrow + \infty).
\end{equation}

Moreover, since $e \geq 2$, we have $\frac{(3p-1)(p^e-1)}{2(p-1)}>\frac{3p^e-1}{2}$, and therefore, by comparing \eqref{eq:asymptoticsamelarge} with \eqref{eq:asympotyclargep^e}, we also conclude that 
\[
S_{3}(p^ek,k,q)=o_k(\lvert {_{p^e}}\mathcal{C}_k \rvert) \ \ \ \mbox{ and } \ \ \ S_{3}(p^ek,k,q)=o_q(\lvert {_{p^e}}\mathcal{C}_k \rvert).
\]

\underline{\textbf{Case 3: $n/k$ is a composite odd number.}} Here, $r=p_z^{e_z}\cdots p_1^{e_1}$, where $p_1>\cdots>p_z$ are odd primes and $z>1$. Since $\ell+1=\left\lceil \frac{p_z^{e_z}\dots p_1^{e_1}}{2}\right\rceil-1$, we have that:
\begin{equation} \label{eq:asympotyclargep_1^e_1p_2...}
S_3(p_z^{e_z}\cdots p_1^{e_1}k,k,q)
%=q^k\frac{q^{\frac{p_1^{e_1}\cdots p_z^{e_z}-1}{2}k}-1}{q^k-1}(q^{p_1^{e_1}\cdots p_z^{e_z}k}-1)
\sim q^{\frac{3p_z^{e_z}\cdots p_1^{e_1}-1}{2}k}, \ \ \  \mbox{ as }k \rightarrow + \infty \ \ \ (\mbox{or as }q \rightarrow + \infty).
\end{equation}

We observe that, since for every $i=1,\dots,z$,
\[
\frac{(3p_i-1)(p_i^{e_i}-1)}{2(p_i-1)} \geq \frac{3p_i^{e_i}-1}{2}
\] it follows that: 
\[
\begin{array}{rl}
     \sum\limits_{i=1}^z\left(\frac{(3p_i-1)(p_i^{e_i}-1)}{2(p_i-1)}\prod\limits_{j=1}^{i-1}p_j^{e_j}k\right) &  \geq   \sum\limits_{i=1}^z\left(\frac{3p_i^{e_i}-1}{2}\prod\limits_{j=1}^{i-1}p_j^{e_j}k\right) \\
      &  =  \frac{1}{2}\sum\limits_{i=1}^z\left(3\prod\limits_{j=1}^{i}p_j^{e_j}-\prod\limits_{j=1}^{i-1}p_j^{e_j}\right)k \\
     & > \frac{3p_z^{e_z}\cdots p_1^{e_1}-1}{2}k.
\end{array}
\]

Therefore, taking into account \eqref{eq:asymptoticdifferentlarge} and \eqref{eq:asympotyclargep_1^e_1p_2...}, we conclude that:  
\[
S_{3}(p_z^{e_z}\cdots p_1^{e_1}k,k,q)=o_k(\lvert _{p_z^{e_z}\dots p_1^{e_1}}\C_k \rvert) \ \ \ \mbox{ and } \ \ \ S_{3}(p_z^{e_z}\cdots p_1^{e_1}k,k,q)=o_q(\lvert _{p_z^{e_z}\dots p_1^{e_1}}\C_k \rvert).
\]
\underline{\textbf{Case 4: $n/k$ is a composite even number.}} Consider the case where $n=rk$, with $r=2^ep_z^{e_z}\cdots p_1^{e_1}$ being even and the product of distinct primes. In this case, by a similar argument of Case 3, we again obtain
\[
S_{3}(2^ep_z^{e_z}\cdots p_1^{e_1},k,q)=o_k(\lvert _{2^ep_z^{e_z}\dots p_1^{e_1}}\C_k \rvert) \ \ \ \mbox{ and } \ \ \ S_{3}(2^ep_z^{e_z}\cdots p_1^{e_1},k,q)=o_q(\lvert _{2^ep_z^{e_z}\dots p_1^{e_1}}\C_k \rvert).
\]
\end{proof}

\begin{remark}
Actually, we have also provided a construction for $r=4$. In this case, the largest known construction is Construction 2) in Table \ref{table:constructionscyclic}. We note that 
\[
S_2(4k,k,q) \sim \frac{1}{2}q^{6k}, \ \ \ \mbox{as } k \rightarrow + \infty \ \ \ (\mbox{or as } q \rightarrow + \infty)
\]
and also, by \eqref{eq:asymptoticRTT+RTT}, that
\[
\lvert {_{2^2}}\mathcal{C}_k \rvert \sim \frac{1}{4}q^{6k}, \ \ \ \mbox{as } k \rightarrow + \infty \ \ \ (\mbox{or as } q \rightarrow + \infty).
\]
Therefore, $_{2^2}\C_k$ has the same asymptotic behavior of $S_2(n,k,q))$, although it should be noted that, in general, $S_2(n,k,q)>\lvert  {_{2^2}}\mathcal{C}_k \rvert $.
\end{remark}

\subsection{Comparison with the Johnson type bound II} 
We compare the size of the codes $_r\C_k$ of Table \ref{table: Our} with the Johnson type bound II as one of the parameters $q$ or $k$ (and so $n=rk$) tends to infinity and the other parameter is treated as constant. The Johnson type bound II for $d=2k-2$ reads as 
\[
\cA_q(n,k,2k-2) \leq J_q(n,2k-2,k):=\left\lfloor\frac{q^n-1}{q^k-1}\left\lfloor\frac{q^{n-1}-1}{q^{k-1}-1}\right\rfloor\right\rfloor.
\]

We also note that when $n=rk$, with $r\geq 2$, we have
\begin{equation} \label{eq:asymtoticanticode}
J_q(rk,2k-2,k) \sim q^{2(r-1)k}, \ \ \ \mbox{as }k \rightarrow + \infty \ \ \ (\mbox{or as }q \rightarrow + \infty ).
\end{equation}

\begin{remark}
Up to now, very few constructions of cyclic subspace codes in $\cG_q(n,k)$ with minimum distance $2k-2$ whose sizes asymptotically reach the Johnson type bound II are known. For $n=2k$, as noted in \cite{roth2017construction}, the size of Construction 1 in \Cref{table:constructionscyclic} approaches the Johnson type bound II within a factor of  $1+o_k(1)$ of the Johnson type bound II as $k$ goes to infinity. Furthermore, it is also easy to verify that the size of Construction 1 also remains within a factor of $1/2+o_q(1)$ of the Johnson type bound II as $q$ goes to infinity. Similarly, for $n=4k$, in \cite{yu2024two} establishes that the size of Construction 2 in \Cref{table:constructionscyclic} remains within a factor of $1/2+o_k(1)$ (of $1/2+o_q(1)$) of the Johnson type bound II as $k$ (as $q$) approaches to infinity.  
Furthermore, the case $n=3k$ is the only one where there exists a family of cyclic subspace codes in $\mathcal{G}_q(3k,k)$ with minimum distance $2k-2$, whose size asymptotically achieves the Johnson type bound II (see \cite[Remark 2]{zhang2024large}).
%it is pointed out that when $n=3k$ the size of Construction 3) in \Cref{table:constructionscyclic} asymptotically approaches the Johnson bound.
Also, we known that for $r \geq 6$ the size of the codes from Constructions 3), 4) and 5) presented in \Cref{table:constructionscyclic}, exhibit the best asymptotic behavior among the constructions for codes in $\mathcal{G}_q(rk,k)$ with minimum distance $2k-2$. However, the sizes of these codes are never asymptotically equal to the Johnson type bound II $J_q(rk,2k-2,k)$. Indeed, taking into account Equation \eqref{eq:asymtoticanticode} and Table \ref{table:constructionscyclic}, we observe that the condition
\[2(r-1)k=(\ell+1)k+n, \ \ \ \mbox{ where } \ell=\lceil r/2 \rceil -2\] is satisfied only when $r=3$. Therefore, $S_3(rk,k,q) \sim J_q(rk,2k-2,k)$ if and only if $r=3$ while for all other values of $r$, we have $S_3(rk,k,q)=o_k(J_q(rk,2k-2,k))$ and $S_3(rk,k,q)=o_q(J_q(rk,2k-2,k))$. 
%Moreover, the case $n=3k$ is the only one where there exists an optimal multi-orbit cyclic subspace code in $\mathcal{G}_q(n,k)$, whose size asymptotically achieves the Johnson bound. 
%We also point out that, as discussed previously for the case $n=2^ek$, there are only two other constructions known to match the asymptotic behavior of the Johnson bound: Construction 1) for $n=2k$ and Construction 2) for $n=4k$ as listed in \Cref{table:constructionscyclic}.
\end{remark}

In what follows, we show that the codes $_{2^e}\C_k$  and $_{3^{e_1}2^e}\C_k$ as in Table \ref{table: Our}, for $n/k\geq 6$ and $q> 2$, have sizes that are within a factor of $1/2^e+o_k(1)$ (and of $1/2^e+o_q(1)$) as $k$ (as $q$) goes to infinity. This proves that the size of the family of codes ${_{2^e}}\mathcal{C}_k$ and $_{3^{e_1}2^e}\C_k$ exhibits the same asymptotic behavior as $J(n,k,q)$ as $q$ or $k$ grows.

\begin{theorem}
    For any $e \geq 2$, the family of codes ${_{2^e}}\mathcal{C}_k \subseteq \mathcal{G}_q(2^ek,k)$ as in \Cref{constr:nestedRRT}, is such that 
    \[
    \lim_{k \rightarrow + \infty}\frac{J(n,k,q)}{\lvert {_{2^e}}\mathcal{C}_k \rvert}=2^e=\lim_{q \rightarrow + \infty}\frac{J(n,k,q)}{\lvert {_{2^e}}\mathcal{C}_k \rvert}.
    \]
\end{theorem}

\begin{proof}
The assertion immediately follows by comparing \eqref{eq:asymptoticRTT+RTT} with \eqref{eq:asymtoticanticode}.
\end{proof}

Now, we consider the case $n=2^e3^{e_1}k$. 

\begin{theorem}
    For any $e,e_1 \geq 1$, the family of codes ${_{2^{e}3^{e_1}}}\mathcal{C}_k \subseteq \mathcal{G}_q(2^{e}3^{e_1}k,k)$ defined as in \Cref{constr:roth+large}, is such that 
    \[
    \lim_{k \rightarrow + \infty}\frac{J(n,k,q)}{\lvert {_{2^{e}3^{e_1}}}\mathcal{C}_k \rvert}=2^{e}=\lim_{q \rightarrow + \infty}\frac{J(n,k,q)}{\lvert {_{2^{e}3^{e_1}}}\mathcal{C}_k \rvert}.
    \]
\end{theorem} 

\begin{proof}
    By \eqref{eq:asymptoticRTT+large}, we know
\[
\lvert _{2^e3^{e_1}}\C_k \rvert \sim \frac{1}{2^e}q^{2\cdot 3^{e_1}k(2^e-1)+2(3^{e_1}-1)k}=\frac{1}{2^e}q^{2(n-k)}, \ \ \ \mbox{ as }k \rightarrow + \infty \ \ \ (\mbox{or as }q \rightarrow + \infty).
\]
Therefore, from \eqref{eq:asymtoticanticode}, the result follows directly.
\end{proof}

Now, we consider the case where $n=3^ek$, with $e\geq 2$. We prove that the size of the code $_{3^e}\C_k$ as in Table \ref{table: Our} asymptotically approaches the Johnson type bound II.

\begin{theorem}
    For any $e \geq 2$, the family of codes ${_{3^e}}\mathcal{C}_k \subseteq \mathcal{G}_q(3^ek,k)$, as defined in \Cref{constr:usinglargenestedpequal} is asymptotically optimal in $k$ and in $q$. 
\end{theorem}

\begin{proof}
From \eqref{eq:asymptoticsamelarge}, we have
\[
\lvert {}_{3^e}\mathcal{C}_k \rvert \sim q^{2(3^e - 1)k} = q^{2(n - k)}, \quad \text{as } k \to +\infty \quad (\text{or as } q \to +\infty).
\]
On the other hand, from \eqref{eq:asymtoticanticode}, we know that \( J(n, k, q) \sim q^{2(n - k)} \). Therefore, the codes \( {}_{3^e}\mathcal{C}_k \) satisfy
\[
\lim_{k \to +\infty} \frac{J(n, k, q)}{\lvert {}_{3^e}\mathcal{C}_k \rvert} = 1 = \lim_{q \to +\infty} \frac{J(n, k, q)}{\lvert {}_{3^e}\mathcal{C}_k \rvert},
\]
from which the result follows directly.

\end{proof}

The above result shows that the family of codes $_{3^e}\C_k$ is asymptotically optimal, achieving the Johnson bound both in the limit as $k\rightarrow +\infty$ and $q\rightarrow +\infty$. Consequently, these codes provide explicit constructions of optimal cyclic subspace codes in $\mathcal{G}_q(n,k)$ , which -for infinitely many values of the ratio $n/k$- asymptotically attain the Johnson bound.

In conclusion, for any integer $r=n/k$, where $r$ is either a power of 2 or a product of positive powers of 2 and 3, we have provided constructions of codes in $\mathcal{G}_q(rk,k)$, for $q> 2$, whose sizes are within a factor of $1/2^e+o_k(1)$ (and of $1/2^e+o_q(1)$) of the Johnson type bound II when $r=2^e$ or $r=2^e3^{e_1}$. In addition, in the case $r=3^e$, the codes $_{3^e}\C_k$ form the first family of cyclic subspace codes in $\mathcal{G}_q(rk,k)$ with minimum distance $2k-2$, whose sizes -for infinitely many values of the integer $r$- asymptotically attain the Johnson type bound II.

\section{Conclusion and Open Problems}
\label{sec:conclusions}
In this article, we introduced a new method for constructing large cyclic subspace codes in $\mathcal{G}_q(rk,k)$ with minimum distance $2k-2$, when $r$ is a composite integer. This operation not only provides new codes with minimum distance $2k-2$ for all parameters considered, but also yields codes with larger size compared to the ones of existing constructions. In addition, we provide the first infinite families of cyclic subspace codes in $\cG_q(rk,k)$ with minimum distance $2k-2$ and sizes that asymptotically approach the Johnson type bound II.
We observe that for values of $r$ different from either $r=2^e$, $e\geq 2$ and $q>2$ or $r=2^e3^{e_1}$, with $e,e_1\geq 1$, the size of the codes $_r\mathcal{C}_k\subseteq\cG_q(rk,k)$ remains asymptotically far from the Johnson type bound II. However, a promising direction for future research is to investigate the existence of cyclic subspace codes in $\mathcal{G}_q(pk, k)$ with minimum distance $2k-2$, for any prime $p \neq 3$, whose sizes asymptotically approach the Johnson type bound II. With the $\odot$ operation introduced in this paper, then it would be possible to obtain, for any composite integer $r$, cyclic subspace codes in $\mathcal{G}_q(rk, k)$ with minimum distance $2k-2$, such that their sizes asymptotically achieve the Johnson type bound II.
Moreover, in this work we have primarily focused on the case where $k \mid n$, which enabled our algebraic constructions through towers of field extensions. However, the case where $k$ does not divide $n$ remains largely unexplored. Finally, we point out that our new constructions of multi-orbit cyclic subspace codes yield large optical orthogonal codes, as shown in \cite{ozbudak2025using}. As a consequence, we immediately obtain new families of optical orthogonal codes with new parameters.

\section*{Acknowledgements}
The research was partially supported by the Italian National Group for Algebraic and Geometric Structures and their Applications (GNSAGA - INdAM).

\bibliographystyle{abbrv}
\bibliography{biblio}

{\small
\noindent Chiara Castello,\\
Dipartimento di Matematica e Fisica,\\ 
Universit\`a degli Studi della Campania ``Luigi Vanvitelli'',\\ 
Viale Lincoln, 5,\\ 
I--\,81100 Caserta, Italy\\
E-mail: chiara.castello@unicampania.it

\noindent Paolo Santonastaso,\\
Dipartimento di Matematica e Fisica,\\ 
Universit\`a degli Studi della Campania ``Luigi Vanvitelli'',\\ 
Viale Lincoln, 5,\\ 
I--\,81100 Caserta, Italy\\
E-mail: paolo.santonastaso@unicampania.it \\
Dipartimento di Meccanica, Matematica e Management, \\
Politecnico di Bari, \\
Via Orabona 4, \\
70125 Bari, Italy \\
E-mail: paolo.santonastaso@poliba.it
}
\end{document}